\documentclass[10pt,draftcls,onecolumn]{IEEEtran}
\ifCLASSINFOpdf
\else
   \usepackage[dvips]{graphicx}
\fi
\usepackage{url}
\usepackage{graphicx}
\usepackage{amsmath}
\usepackage{comment}
\usepackage[utf8]{inputenc}
\usepackage[english]{babel}
\usepackage{color}
\usepackage{framed}
\usepackage{amsmath,amssymb,amsfonts}

\usepackage{graphicx}
\usepackage{upgreek}
\usepackage{textcomp}
\usepackage{hyperref}
\hypersetup{
    colorlinks=true,
    linkcolor=blue,
    filecolor=magenta,      
    urlcolor=cyan,
    citecolor=blue,
    }
\usepackage{url}

\usepackage{cite}

\usepackage[dvipsnames]{xcolor}
\usepackage{amsmath}
\usepackage{amsthm}
\usepackage{dsfont}
\usepackage{thmtools}
\usepackage{amssymb}

\def\btau{\boldsymbol{\tau}}

\def\bpsi{\boldsymbol{\psi}}

\def\bOmega{\boldsymbol{\Omega}}

\def\mbx{\mathbf{x}}
\def\mby{\mathbf{y}}

\def\mbA{\mathbf{A}}

\def\mbE{\mathbf{E}}

\def\mbM{\mathbf{M}}
\def\mbN{\mathbf{N}}

\def\mbP{\mathbf{P}}
\def\mbQ{\mathbf{Q}}
\def\mbR{\mathbf{R}}

\def\mbT{\mathbf{T}}
\def\mbU{\mathbf{U}}

\def\mbX{\mathbf{X}}
\def\mbY{\mathbf{Y}}
\def\mbZ{\mathbf{Z}}

\newtheorem{theorem}{Theorem}
\newtheorem{proposition}{Proposition}
\newtheorem{lemma}{Lemma}
\newtheorem{corollary}{Corollary}

\theoremstyle{definition}
\newtheorem{definition}{Definition}

\usepackage{ifpdf}
\usepackage{times}
\usepackage{layout}
\usepackage{float}
\usepackage{afterpage}
\usepackage{amsmath}
\usepackage{amstext}
\usepackage{amssymb,bm}
\usepackage{latexsym}
\usepackage{color}
\usepackage{flexisym}
\usepackage{kantlipsum}
\usepackage{graphicx}
\usepackage{amsmath}
\usepackage{amsthm}
\usepackage{graphicx}
\usepackage{cleveref}
\usepackage[inkscapeformat=png]{svg}
\usepackage[caption=false,font=footnotesize]{subfig}

\usepackage{booktabs}
\usepackage{multicol}

\usepackage{enumitem}

\usepackage[switch]{lineno}
\usepackage[named]{algo}
\usepackage[noend]{algpseudocode}
\usepackage{algorithm}

\makeatletter
\newcommand*{\rom}[1]{\expandafter\@slowromancap\romannumeral #1@}
\makeatother

\begin{document}
\setlength{\abovedisplayskip}{3pt}
\setlength{\belowdisplayskip}{3pt}

\title{From Quasi-Isometric Embeddings to Finite-Volume Property:\\ A Theoretical Framework for Quantized Matrix Completion
}

\author{Arian Eamaz, Farhang Yeganegi, and Mojtaba Soltanalian

\thanks{The authors are with the Department of Electrical and Computer Engineering, University of Illinois Chicago, Chicago, IL 60607, USA (e-mail: \emph{ \{aeamaz2, fyegan2, msol\}@uic.edu}). The first two authors contributed equally to this work.}
}
\maketitle

\begin{abstract}
We delve into the impact of \emph{memoryless scalar quantization} on matrix completion. Our primary motivation for this research is to evaluate the recovery performance of nuclear norm minimization in handling quantized matrix problems without the use of any regularization terms such as those stemming from maximum likelihood estimation. We broaden our theoretical discussion to encompass the coarse quantization scenario with a dithering scheme, where the only available information for low-rank matrix recovery is a few-bit low-resolution data.
We furnish theoretical guarantees for both scenarios: when access to dithers is available during the reconstruction process, and when we have access solely to the statistical properties of the dithers. Additionally, we conduct a comprehensive analysis of the effects of sign flips and prequantization noise on the recovery performance, particularly when the impact of sign flips is quantified using the well-known Hamming distance in the upper bound of recovery error.
\end{abstract}


\setlength{\abovedisplayskip}{3pt}
\setlength{\belowdisplayskip}{3pt}

\section{Introduction}
\label{intro}
Matrix completion, the task of reconstructing an unknown low-rank matrix from partial data, presents a pervasive challenge in various practical domains, including collaborative filtering \cite{goldberg1992using}, system identification \cite{liu2010interior}, and sensor localization \cite{singer2008remark}. A significant insight has emerged: in the case of a rank-$r$ matrix $\mbX$, and when a certain level of ``structure'' is absent, a sparse, randomly selected subset of its elements can enable precise reconstruction. This groundbreaking discovery was initially validated by \cite{candes2009exact}, which meticulously analyzed a convex relaxation method introduced by \cite{fazel2002matrix}. An essential determinant of the feasibility of matrix completion is a specific incoherence measure, introduced in \cite{candes2009exact}. 

Quantization is a fundamental process in digital signal processing, converting continuous signals into discrete representations. However, achieving high-resolution quantization often demands a large number of quantization levels, which can lead to increased power consumption, higher manufacturing costs, and reduced sampling rates in analog-to-digital converters (ADCs). To address these challenges, researchers have explored the use of reduced quantization bits, including the extreme case of one-bit quantization, where signals are compared with a fixed threshold at the ADCs, yielding binary outputs \cite{jacques2013robust,boufounos2015quantization}. This approach allows for high-rate sampling while decreasing implementation costs and energy consumption compared to multi-bit ADCs. One-bit ADCs have proven highly valuable in various applications, including MIMO systems \cite{kong2018multipair,mezghani2018blind}, channel estimation \cite{li2017channel}, and array signal processing \cite{liu2017one}.

Scalar quantization with \emph{dithering} is a technique that involves intentionally introducing random noise to an input signal before quantization \cite{gray2002dithered}. This method has a well-established presence in both practical applications, where it can lead to more enhanced reconstructions, and theoretical contexts, where it often yields favorable statistical properties of the quantization noise. This approach is widely acknowledged and referenced in the literature \cite{schuchman1964dither,vanderkooy1987dither,wagdy1989validity,carbone1997quantitative,ali202012}. More recently, dithered quantization has found applications in the realm of high-dimensional structured signal recovery from quantized linear measurements, as demonstrated in various domains, including sparse parameter estimation \cite{dabeer2006signal,thrampoulidis2020generalized}, compressed sensing \cite{xu2020quantized,dirksen2021non,eamaz2024harnessing,dirksen2022sharp,dirksen2020one}, phase retrieval \cite{eamaz2022phase}, covariance recovery \cite{dirksen2022covariance,eamaz2021modified,eamaz2023covariance,eamaz2022covariance,dirksen2023tuning,chen2023high,chen2023parameter} and sampling theory \cite{eamaz2022uno}.

\subsection{Related Works}
There is a rich body of literature addressing various inverse problems under quantized measurements, most notably in the context of compressed sensing. These studies primarily focus on establishing theoretical guarantees that provide insights into the recovery error and the required number of samples. In \cite{plan2014dimension}, the authors introduced the concept of \emph{random hyperplane tessellations} to build a probabilistic embedding between the Hamming distance and the directional recovery error. Subsequently, \cite{oymak2015near} characterized the trade-off between distortion and sample complexity in random hyperplane tessellations using Gaussian complexity, addressing both arbitrary and structured signal sets. However, these approaches suffer from two major limitations: they rely on a ditherless quantization scheme limited to direction-only estimation, and they assume Gaussian measurement matrices exclusively.  

Recent works have demonstrated that, under suitable conditions, complete signal reconstruction is possible by introducing nonzero random thresholds or dithering during quantization \cite{knudson2016one,baraniuk2017exponential,xu2020quantized}. The use of uniform dithering has opened new theoretical and practical opportunities, effectively extending the direction-only estimation framework of the ditherless setup to one that also allows amplitude estimation and a more accurate recovery of the original signal. For instance, \cite{jacques2017small} proposed quasi-isometric embeddings achieved with high probability through scalar (dithered) quantization following a linear random projection. In such embeddings, both multiplicative and additive distortions coexist when distances between mapped vectors are measured using the $\ell_1$-norm.  

In the Gaussian setting, \cite{baraniuk2017exponential} suggested incorporating either adaptive or random dithers before binarization of the compressive measurements. More recently, \cite{xu2020quantized} extended these ideas beyond Gaussian measurements, introducing a scalar quantization scheme with uniformly distributed dithering. Their analysis provided theoretical guarantees for measurement matrices satisfying the restricted isometry property (RIP) and, notably, did not rely on the high-resolution assumption, making the results valid for a wide range of bit-rates. In \cite{derezinski2022sharp}, random hyperplane tessellations were further refined to establish embeddings between the Hamming distance and the recovery error in the presence of uniform dithering, though their analysis was limited to subgaussian matrices and lacked guarantees for deterministic measurements.  

In another line of work, \cite{thrampoulidis2020generalized} addressed parameter estimation under dithered quantization by reformulating the problem as a generalized LASSO, thereby leveraging existing theoretical guarantees from that framework. In the broader context, the search for Johnson–Lindenstrauss (JL) style isometric mappings has been central to understanding how random sensing matrices, such as subgaussian matrices with unit variance, can preserve pairwise distances between data points up to a multiplicative distortion \cite{dasgupta2003elementary,vershynin2010introduction}. In other words, for random matrix sensing $\mbA\in\mathbb{R}^{m\times n}$ and all $(\mbx,\mby)\in\mathcal{K}$, there exists an $\varepsilon$-isometry between two metric spaces $\left(\mathcal{K},\ell_p\right)$ and $\left(\mbA\mathcal{K},\ell_q\right)$ with a high probability
\begin{equation}
\alpha(\varepsilon)\left\|\mbx-\mby\right\|_p \leq\left\|\mbA\mbx-\mbA\mby\right\|_q\leq \beta(\varepsilon)\left\|\mbx-\mby\right\|_p,
\end{equation}
where $\alpha(\cdot)$ and $\beta(\cdot)$ are functions depending on $m$. 

One well-known example of such an isometry is the RIP \cite{candes2008introduction}, which ensures the existence of accurate recovery solutions for sparse signals. The study of these isometries across various random mappings has provided valuable insights into the existence and stability of solutions, the recovery error, and the sample complexity required for reliable reconstruction. Motivated by these developments, similar efforts have been made in the context of quantized measurements, aiming to identify analogous isometric structures within quantized settings. In the ditherless scenario, a quasi-isometric form can be established; however, it includes an additional constant $\Delta$ in the upper bound, which is ideally expected to be controllable through the number of samples \cite{boufounos2013efficient,boufounos2015quantization,jacques2015quantized}:
\begin{equation}
(1-\varepsilon)\left\|\mbx-\mby\right\|_2-\Delta\leq\frac{1}{\sqrt{m}}\left\|\mathcal{Q}_{\Delta}(\mbA\mbx)-\mathcal{Q}_{\Delta}(\mbA\mby)\right\|_2\leq (1+\varepsilon)\left\|\mbx-\mby\right\|_2+\Delta,
\end{equation}
with $\mathcal{Q}_{\Delta}(\cdot)$ being scalar quantizer. Since the quantizer is a discontinuous function, quantization introduces abrupt jumps whenever a projected value crosses a quantization threshold. To analyze these discontinuities, one can either employ a softening strategy \cite{plan2014dimension,jacques2017time,jacques2017small} for the quantization operator or bound the number of discontinuous components using deterministic properties of the quantizer \cite{boufounos2017representation,xu2020quantized}.

In particular, when the quantization process incorporates uniform dithering and exploits its inherent statistical properties, for subgaussian sensing measurements and constants $\kappa$, $K$, and $c$, the following quasi-isometric embedding holds with a high probability \cite{jacques2017small}:
\begin{equation}
\label{jaack}
\begin{aligned}
\left|\frac{1}{m}\left\|\mathcal{Q}_{\Delta,\tau}(\mbA\mbx)-\mathcal{Q}_{\Delta,\tau}(\mbA\mby)\right\|_1-\left(\frac{2}{\pi}\right)^{\frac{1}{2}}\left\|\mbx-\mby\right\|_2\right|
\leq\left(\varepsilon+\frac{\kappa}{\sqrt{K}}\right)\left\|\mbx-\mby\right\|_2+c \varepsilon \Delta,
\end{aligned}
\end{equation}
where $\mathcal{Q}_{\Delta,\tau}(\cdot)$ is the dithered quantizer with uniform dithers $\tau$. As can be observed, due to the combined effects of subgaussian measurements and uniform dithering, the JL-style embedding transitions from an $\ell_2 / \ell_2$ form to an $\ell_1 / \ell_2$ embedding. In this case, the structure no longer exhibits a strict $\varepsilon$-isometry, as the differing statistical properties of the measurements inherently influence the form of the resulting quasi-isometry.

More than isometry-style theoretical guarantees, in \cite{eamaz2024harnessing}, we introduced the Finite Volume Property (FVP) as a new framework to analyze one-bit sensing with dithered quantization. The FVP characterizes how the collection of one-bit inequalities forms a finite-volume polyhedron enclosing the true signal. As the number of samples increases, the average distance between the signal and these hyperplanes converges to its mean, and the finite volume around the signal shrinks, ensuring that the reconstructed solution lies within a small ball centered at the true signal. Mathematically, the FVP is formalized through a concentration inequality on this average distance, establishing a quasi-isometric embedding between the signal space and its one-bit measurements under isotropic sampling matrices. This result provides an explicit upper bound on the recovery error and determines the number of one-bit samples required for accurate reconstruction. The importance of the FVP lies in its conceptual and practical impact: unlike traditional random-hyperplane tessellation results, it views one-bit sensing as a linear feasibility problem and yields uniform recovery guarantees that extend beyond random or Gaussian sampling to include deterministic matrices. 

One of the less explored areas in the theoretical one-bit sensing is \emph{quantized matrix completion}, where the objective is to recover a low-rank matrix from quantized observations. Initial attempts to address quantized matrix completion can be found in \cite{davenport20141} and \cite{bhaskar20151}, where researchers developed theoretical guarantees within the framework of the generalized linear model. They derived a regularized maximum likelihood estimate (MLE) based on a probability distribution determined by the real-valued noisy entries of the low-rank matrix. To regularize the MLE problem, these studies employed both nuclear and Frobenius norms, drawing inspiration from prior work on one-bit compressed sensing \cite{davenport20141}. The optimization of the regularized MLE was carried out using projected gradient descent.

The results of \cite{davenport20141} and \cite{bhaskar20151} presented various theoretical guarantees and necessary conditions for achieving perfect recovery performance in the context of the MLE problem. These guarantees were established by considering the rank and the number of measurements required, ensuring optimal solutions when the density function is concave. In the realm of the one-bit matrix completion problem, a comprehensive investigation of a max-norm constrained MLE was conducted in \cite{cai2013max}. Further advancements were made in \cite{ni2016optimal}, where the authors developed a greedy algorithm that extended the concept of conditional gradient descent to efficiently solve the regularized MLE for the one-bit matrix completion problem. The concept of regularized MLE for the one-bit matrix completion problem was extended to the quantized matrix completion problem in \cite{bhaskar2016probabilistic}. Here, the authors of \cite{cao2015categorical} considered a trace-norm regularized MLE with a likelihood function for categorical distributions. In \cite{gao2018low}, a regularized MLE for matrix completion from quantized and erroneous measurements was proposed, accounting for the presence of sparse additive errors in the model.

\subsection{Motivations}
In \cite{davenport20141}, the authors derived the MLE while incorporating time-varying thresholds, which correspond to random dithering, in the context of noisy measurements. However, it is crucial to recognize that the design of these time-varying thresholds plays a pivotal role in one-bit sensing, as highlighted in \cite{baraniuk2017exponential,eamaz2022uno,eamaz2022phase,eamaz2022covariance, xu2020quantized}, and can significantly enhance signal reconstruction performance. Nonetheless, relying on noise as our source of dithering, as demonstrated in \cite{davenport20141}, confines us to thresholds that mimic the behavior of the noise. Furthermore, when dealing with non-convex distributions, the uniqueness of the solution in the MLE problem cannot be guaranteed. To make matrix recovery feasible, certain assumptions must be imposed on the noise distribution, as discussed in \cite{davenport20141, cai2013max, bhaskar20151, gao2018low, ni2016optimal}. In practical scenarios, it is often unrealistic to assume that the noise adheres to a distribution with characteristics that guarantee efficient recovery in the context of the regularized MLE problem. In fact, in many real-world scenarios, the distribution of noise remains unknown.

This paper is motivated by the desire to investigate the impact of scalar quantization broadly defined on low-rank matrix completion, without being restricted to any specific problem formulation. Our goal is to derive an isometric embedding for this setting that holds independently of the particular reconstruction approach or algorithm employed, such as nuclear norm minimization or its regularized variants (e.g., MLE or other forms of additional regularization).

\subsection{Contribution}
In this paper, we propose a quasi-isometric embedding for matrix completion under uniform sampling and dithered quantization (Proposition~1). Building on this result, we derive a concrete upper bound on the recovery error and establish its decay rate with respect to the number of available samples (Theorem~1). Furthermore, we develop FVP-style guarantees for one-bit matrix completion, where the one-bit samples define a linear feasible system with a mean-square-error–like criterion. This formulation provides a continuous objective that avoids the softening or non-differentiability issues inherent to discontinuous quantizers. The corresponding concentration results and proofs are considerably simpler and are presented in Theorem~3 and Theorem~4.

\subsection{Notation}
Throughout this paper, we use bold lowercase and bold uppercase letters for vectors and matrices, respectively.  We represent a vector $\mathbf{x}$ and a matrix $\mbX$ in terms of their elements as $\mathbf{x}=[x_{i}]$ and $\mathbf{X}=[X_{i,j}]$, respectively. The set of real numbers is $\mathbb{R}$. For vectors, we define $\mathbf{x}\succeq \mathbf{y}$ as a component-wise inequality between vectors $\mathbf{x}$ and $\mathbf{y}$, i.e., $x_{i}\geq y_{i}$ for every index $i$. For matrices, $\mbX\succeq\mbY$ implies that $\mbX-\mbY$ is a positive semi-definite matrix. The function $\textrm{diag}(.)$ returns the diagonal elements of the input matrix. The nuclear norm of a matrix $\mbX\in \mathbb{R}^{n_1\times n_2}$ is denoted $\left\|\mbX\right\|_{\star}=\sum^{r}_{i=1}\sigma_{i}$ where $r$ and $\left\{\sigma_{i}\right\}$ are the rank and singular values of $\mbX$, respectively. The Frobenius norm of a matrix $\mathbf{X}\in \mathbb{R}^{n_1\times n_2}$ is defined as $\|\mathbf{X}\|_{\mathrm{F}}=\sqrt{\sum^{n_1}_{r=1}\sum^{n_2}_{s=1}\left|x_{rs}\right|^{2}}$, where $x_{rs}$ is the $(r,s)$-th entry of $\mathbf{X}$. We also define $\|\mbX\|_{\mathrm{max}}=\sup_{i,j}|X_{i,j}|$. The $\ell_{p}$-norm of a vector $\mathbf{x}$ is $\|\mathbf{x}\|_{p}=\left(\sum_{i}x^{p}_{i}\right)^{1/p}$. 
The Hadamard (element-wise) products is $\odot$. The diameter of a bounded set $\mathcal{K}\subset \mathbb{R}^{n}$ is written as $\left\|\mathcal{K}\right\|=\sup_{\mbx\in\mathcal{K}} \|\mbx\|_2$. The vectorized form of a matrix $\mbX$ is written as $\operatorname{vec}(\mbX)$. The $\ell_1$-norm for a matrix $\mbX$ means $\|\mbX\|_1=\|\operatorname{vec}(\mbX)\|_1$. For an event $\mathcal{E}$, $\mathbb{I}[\mathcal{E}]$ is the indicator function for that event meaning that $\mathbb{I}[\mathcal{E}]$ is $1$ if $\mathcal{E}$ occurs; otherwise, it is zero. The set $[n]$ is defined as $[n]=\left\{1,\cdots,n\right\}$. The function $\operatorname{sgn}(\cdot)$ yields the sign of its argument. The Hamming distance between $\operatorname{sgn}(\mbx),\operatorname{sgn}(\mby)\in\{-1,1\}^n$ 
is defined as 
\begin{equation}
\label{hamming}
d_{\mathrm{H}}(\operatorname{sgn}(\mbx),\operatorname{sgn}(\mby))=\sum_{i=1}^{n}\mathbb{I}_{(\operatorname{sgn}(x_i)\neq \operatorname{sgn}(y_i))}.
\end{equation}
The function $\log(\cdot)$ denotes the natural logarithm, unless its base is otherwise stated. The notation $x \sim \mathcal{U}_[a,b]$ means a random variable drawn from the uniform distribution over the interval $[a,b]$. The Kolmogorov $r$-entropy of a set $\mathcal{K}$ is denoted by $\mathcal{H}\left(\mathcal{K},r\right)$ defined as
the logarithm of the size of the smallest $r$-net of $\mathcal{K}$ \cite{kolmogorov1959varepsilon}. The subgaussian norm of a random variable $X$ is characterized by 
\begin{equation}
\|X\|_{\psi_2}=\inf \left\{t>0: \mathbb{E}e^{X^2/t^2} \leq 2\right\}
\end{equation}
The sub-exponential norm of a random variable $X$ is characterized by
\begin{equation}
\|X\|_{\psi_1}=\inf \{t>0: \mathbb{E} e^{|X| / t} \leq 2\}.
\end{equation}

\section{Quasi-Isometry for Quantized Matrix Completion}
\label{OB-MC}
This section commences with an introduction to scalar quantization and its established variants in the literature. Subsequently, we delve into a crucial property of uniform quantization (scalar quantization with uniform dithering), which plays a pivotal role in our theoretical guarantees. Finally, we present the guarantees for quantized matrix completion.

\subsection{Scalar Quantization}
\label{ss_1}
The memoryless scalar quantizer
\begin{equation}
\mathcal{Q}_{\Delta}: \mathbb{R} \rightarrow \mathcal{A}_K.
\end{equation}
is defined as
\begin{equation}
\label{a4}
\mathcal{Q}_{\Delta}(x)=\Delta\left(\left\lfloor\frac{x}{\Delta}\right\rfloor+\frac{1}{2}\right),
\end{equation}
where $\Delta$ is resolution parameter and $\mathcal{A}_K$ is the finite alphabet set given by
\begin{equation}
\label{a2}
\mathcal{A}_K:=\left\{ \pm  \frac{k \Delta}{2}: 0 \leq k \leq K, k \in \mathbb{Z}\right\}.
\end{equation}
When we introduce a uniform dither generated as $\tau\sim \mathcal{U}_{\left[-\frac{\Delta}{2},\frac{\Delta}{2}\right]}$, to the input signal of the quantizer, the resulting quantization process is termed \emph{uniform quantization}. This process can be defined as follows:
\begin{equation}
\label{a1}
\mathcal{Q}_{\Delta, \tau}\left(x\right)=\Delta\left(\left\lfloor\frac{x+\tau}{\Delta}\right\rfloor+\frac{1}{2}\right).
\end{equation}
Random dithering is realized through a randomly (usually, Gaussian and Uniform) dithered generator within the
ADC \cite{robinson2019analog}. The source of this Uniform dither is a low-cost thermal noise diode, which may require additional circuitry and amplifiers to enhance the noise levels; see, for instance, \cite{ali2020background} for the implementation of multiple dithering in a 12-bit, 18 gigasamples per second (GS/s) ADC.

When quantizing a scalar $x$, it is essential to recognize that the uniform quantizer effectively becomes a 1-bit quantizer (scaled appropriately) when the resolution parameters exceed the magnitude of the signal: 
\begin{equation}
\label{a5}
\mathcal{Q}_{\Delta}\left(x\right)=\frac{\Delta}{2} \operatorname{sgn}(x), \quad |x|<\Delta.
\end{equation}
This remains true when the quantizer is associated with a uniform dither
\begin{equation}
\label{a7}
\mathcal{Q}_{\Delta, \tau}\left(x\right)=\frac{\Delta}{2} \operatorname{sgn}(x+\tau), \quad|x|<\frac{\Delta}{2}. 
\end{equation}
One intriguing property of the uniform quantizer is its ability to offset the quantization impact on average, as articulated in the following lemma:
\begin{lemma}[\cite{thrampoulidis2020generalized}]
\label{lemma_1}
Let $\tau$ be a random dither distributed according to $\tau\sim \mathcal{U}_{\left[-\frac{\Delta}{2},\frac{\Delta}{2}\right]}$. Then for a fixed $x\in \mathbb{R}$, we have    
\begin{equation}
\label{a8}   
\mathbb{E}\mathcal{Q}_{\Delta, \tau}\left(x\right) = x.
\end{equation}
\end{lemma}

\subsection{Quantized Matrix Completion}
\label{ss_2}
Assume we apply the coarse quantization to the observed partial entries of a low-rank matrix 
$\mbX\in\mathbb{R}^{n_1\times n_2}$ of rank $r$. Define $\mathcal{P}_{\Omega}\left(\mbX\right)=\left[X_{i,j}\right]$ to be the orthogonal projector onto the span of matrices vanishing outside of $\Omega$ with the cardinality $m$. In quantized matrix completion, we solely observe the partial matrix through the quantized data as below:
\begin{equation}
\label{a6}
\begin{aligned}
Q_{i,j} &= \begin{cases} \mathcal{Q}_{\Delta}\left(X_{i,j}\right) & (i,j)\in \Omega,\\ 0 & \text{otherwise}.
\end{cases}
\end{aligned}
\end{equation}
In our setting, similar to previous works  \cite{candes2010matrix,candes2011tight,davenport20141}, we employ uniform sampling for matrix completion. Sensing and reconstructing matrices from a limited number of noisy entries is a captivating and ongoing area of research that has garnered significant attention. In the work presented by the authors in \cite{candes2010matrix}, 
noisy matrix completion is formulated as a 
nuclear norm minimization problem. This approach has led to the derivation of rigorous theoretical guarantees, further advancing the understanding and development of this intriguing field. Consider the noisy measurements as follows:
\begin{equation}
X^{(n)}_{i,j} = X_{i,j}+Z_{i,j},\quad (i,j)\in \Omega,
\end{equation}
where $Z_{i,j}$ is a bounded additive noise. 

Extensive investigations conducted in \cite{candes2010matrix,cai2010singular} have demonstrated that matrix completion with noise can be formulated as a nuclear norm minimization problem as follows:
\begin{equation}
\label{Steph1}
\begin{aligned}
&\underset{\mbX}{\textrm{minimize}}\quad \left\|\mbX\right\|_{\star}\\
&\text{subject to} \quad \left\|\mathcal{P}_{\Omega}\left(\mbX-\mbX^{(n)}\right)\right\|_{\mathrm{F}} \leq \delta,
\end{aligned}
\end{equation}
where $\mbX^{(n)}$ is the noisy matrix and $\delta$ presents the effect of noise. Drawing inspiration from the theoretical guarantees of noisy matrix completion, we examine the problem of quantized matrix completion, specifically focusing on the case of one-bit matrix completion. Let us assume that $\mbQ = \mathcal{Q}_{\Delta}\left(\mathcal{P}_{\Omega}\left(\mbX\right)\right)\in\mathcal{A}^{n_1\times n_2}_{K}$ represents a scalar quantization of known entries of low-rank matrix $\mbX$, where only entries of $(i,j)\in\Omega$ are quantized, and the remaining entries become zero. Consequently, the quantized measurements can be expressed as follows:
\begin{equation}
\label{Steph2}
\mbQ = \mathcal{P}_{\Omega}\left(\mbX\right)+\mbN,
\end{equation}
where the matrix $\mbN\in\mathbb{R}^{n_1\times n_2}$ presents the effect of quantization as the additive noise matrix. Therefore, the nuclear norm minimization problem 
associated with the quantized matrix completion is given by
\begin{equation}
\label{Steph3}
\begin{aligned}
&\underset{\mbX}{\textrm{minimize}}\quad \left\|\mbX\right\|_{\star}\\
&\text{subject to} \quad \left\|\mathcal{P}_{\Omega}(\mbX)-\mbQ\right\|_{\mathrm{F}} \leq \delta,
\end{aligned}
\end{equation}
where the parameter $\delta$ denotes the impact of the quantization process.

Drawing inspiration from \cite[Theorem~7]{candes2010matrix} and \cite{candes2011tight} in the context of noisy matrix completion problem, we can derive an upper bound for the Frobenius norm error in quantized matrix completion through the nuclear norm minimization problem. This result is stated in the following theorem:
\begin{theorem}\cite[Theorem~7]{candes2010matrix}
\label{St_theorem1}
For the uniform set $\Omega$ and $\delta$, a parameter presenting the effect of quantization, define $m$ as the cardinality of the set $\Omega$. The error norm $\left\|\mbX-\bar{\mbX}\right\|_{\mathrm{F}}$ between the fixed matrix $\mbX$ and the reconstructed matrix $\bar{\mbX}$ by the quantized matrix completion is bounded with a probability of at least $1-\left(\max \left(n_1, n_2\right)\right)^{-3}$ as  
\begin{equation}
\left\|\mbX-\bar{\mbX}\right\|_{\mathrm{F}} \leq 4 \sqrt{\frac{\left(2 n_1 n_2+m\right) \min \left(n_1, n_2\right)}{m}} \delta+2 \delta.
\end{equation}
\end{theorem}
However, this bound for the quantized measurements presents several limitations that motivate us to explore alternative approaches to this problem. First, in recovery problems, we typically seek bounds that provide a meaningful relationship with the number of samples, as such bounds help determine the required sample size to achieve acceptable reconstruction performance. In this case, when adapting the noisy matrix completion problem to the quantized setting, it becomes evident that the upper bound 
$\delta$ is directly related to the number of samples. This implies that increasing the number of samples may actually lead to a deterioration in the upper bound of the recovery error. We seek a bound that improves as the number of samples increases, which is a natural and desirable property in coarse quantization scenarios where sample abundance can enhance recovery performance \cite{eamaz2024harnessing, eamaz2022phase, eamaz2023one}.

Moreover, the existing bound is derived for a fixed low-rank matrix, whereas in quantization schemes that leverage uniform dithering, we have this opportunity to establish universal results that hold for all low-rank matrices.

All these considerations motivate us to pursue alternative approaches from the literature, such as isometric embeddings, which leverage the randomness introduced by uniform dithering. The unique statistical properties of uniform dithers can yield remarkably strong results, as explored in the following section.
\subsection{Quasi-Isometric Quantized Embedding}
\label{ss_3}
Let $\mathcal{K}^r$ denote the set of low-rank matrices. In this section, we establish an embedding between the metric spaces $\left(\mathcal{K}^r\subset \mathbb{R}^{n_1\times n_2},\ell_1\right)$ and $\left(\mathcal{Q}\left(\mathcal{P}_{\Omega}\left(\mathcal{K}^r\right)\right)\subset \mathcal{A}^{n_1\times n_2}_{K}, \ell_1\right)$, where $\mathcal{Q}\left(\mathcal{P}_{\Omega}\left(\mathcal{K}^r\right)\right)$ represents the space of quantized matrix completion. This embedding aids in achieving the recovery performance in quantized matrix completion with high probability. To establish this embedding, we employ the following analytical tools:
\begin{definition}
\label{def_1}
Define a low-rank matrix as $\mbX=[X_{i,j}]\in\mathbb{R}^{n_1\times n_2}$ and the dither matrix by $\mbT=\left[\tau_{i,j}\right]\in \mathbb{R}^{n_1\times n_2}$. The consistency property of uniform quantization over the pair $\mbX,\mbY\in\mathcal{K}^r$, is given by
\begin{equation}
\label{a_1}
\mathcal{Q}_{\Delta, \mbT}\left(\mathcal{P}_{\Omega}\left(\mbX\right)\right)=\mathcal{Q}_{\Delta, \mbT}\left(\mathcal{P}_{\Omega}\left(\mbY\right)\right).
\end{equation}
\end{definition}
The concept of consistent reconstruction, as defined in Definition~\ref{def_1}, has played a pivotal role in obtaining theoretical guarantees in the field of one-bit compressed sensing \cite{jacques2013robust, xu2020quantized, eamaz2024harnessing, baraniuk2017exponential} and one-bit low-rank matrix sensing \cite{eamaz2024harnessing, foucart2019recovering}, as discussed in the review of prior literature.

For $(\mbX,\mbY)\in\mathcal{K}^r$, the distance used in the embedding $\ell_1$ is defined as
\begin{equation}
\label{n1} 
\mathcal{D}\left(\mbX,\mbY\right) \triangleq \frac{1}{m} \left\|\mathcal{Q}_{\Delta, \mbT}\left(\mathcal{P}_{\Omega}\left(\mbX\right)\right)-\mathcal{Q}_{\Delta, \mbT}\left(\mathcal{P}_{\Omega}\left(\mbY\right)\right)\right\|_1.
\end{equation}
An interesting aspect of the $\ell_1$ distance is that when we employ uniform dithering in our scheme, the impact of quantization can be mitigated by taking the expectation, as shown in the following lemma \cite{xu2020quantized}:
\begin{lemma}
\label{Jack}
For real values $a, b\in\mathbb{R}$ and $\tau\sim\mathcal{U}_{[-\frac{\Delta}{2},\frac{\Delta}{2}]}$, we have
\begin{equation}
\mathbb{E}_{\tau}\left|\mathcal{Q}(a+\tau)-\mathcal{Q}(b+\tau)\right| = \left|a-b\right|.
\end{equation}
\end{lemma}
Thus, we can evaluate the mean of the distance as stated in this lemma:
\begin{lemma}
\label{yuta}
If we have the uniform sampling in the matrix completion problem, for a pair $\left(\mbX,\mbY\right)\in\mathcal{K}^r$, and $[\mbT]_{i,j}\sim\mathcal{U}_{[-\frac{\Delta}{2},\frac{\Delta}{2}]}$, the following relation is obtained:
\begin{equation}
\mathbb{E}_{\tau,(i,j)}\left\|\mathcal{Q}_{\Delta, \mbT}\left(\mathcal{P}_{\Omega}\left(\mbX\right)\right)-\mathcal{Q}_{\Delta, \mbT}\left(\mathcal{P}_{\Omega}\left(\mbY\right)\right)\right\|_1 = \frac{m}{n_1n_2}\left\|\mbX-\mbY\right\|_1.
\end{equation}
\end{lemma}
\begin{IEEEproof}
The proof of Lemma~\ref{yuta}, is straightforward by considering Lemma~\ref{Jack} for the expectation over the dither values. Due to the fact that we use the uniform sampling in the matrix completion, the expectation over the indices $(i,j)\in\Omega$ is readily given by
\begin{equation}
\mathbb{E}_{(i,j)}\sum_{(i,j)\in\Omega}\left|X_{i,j}-Y_{i,j}\right| = \sum_{(i,j)\in\Omega}\sum_{(i,j)\in[n_1]\times[n_2]}\frac{1}{n_1n_2}\left|X_{i,j}-Y_{i,j}\right| = \frac{m}{n_1n_2}\left\|\mbX-\mbY\right\|_1.  
\end{equation}
\end{IEEEproof}
Using uniform sampling in matrix completion and integrating uniform dithering into the quantization process, we find that taking the expectation effectively cancels out the effects of both the sensing process and the quantization. This finding shows that the quantized data retains the statistical information from the high-resolution measurements.

We can express the $\ell_1$ distance between quantized values using the following indicator function:
\begin{equation}
\left|\mathcal{Q}\left(a\right)-\mathcal{Q}\left(b\right)\right| =  \Delta\sum_{k\in\mathbb{Z}} \mathbb{I}\left[\mathcal{E}(a-k\Delta,b-k\Delta)\right],  
\end{equation}
where 
\begin{equation}
\mathcal{E}(a, b):=\{\operatorname{sgn}(a) \neq \operatorname{sgn}(b)\}.
\end{equation}
Authors of \cite{jacques2017small} show that the $\ell_1$ distance can be smoothed by introducing $t$:
\begin{equation}
\label{e1}
d^t(a, b):=\Delta \sum_{k \in \mathbb{Z}} \mathbb{I}\left[\mathcal{F}^t(a-k \Delta, b-k \Delta)\right] \in \Delta \mathbb{N}.
\end{equation}
where
\begin{equation}
\label{def_f_t}
\mathcal{F}^t(a-k\Delta,b-k\Delta)\triangleq \left\{a-k\Delta>t, b-k\Delta\leq-t\right\}\cup\left\{a-k\Delta<-t, b-k\Delta\geq t\right\}.
\end{equation}
Therefore, we have
\begin{equation}
\mathcal{D}^t(\mbX, \mbY)=\frac{1}{m} \sum_{(i,j)\in\Omega} d^t\left(X_{i,j}+\tau_{i,j}, Y_{i,j}+\tau_{i,j}\right).
\end{equation}
The pseudo-distance $\mathcal{D}^t$
can both upper and lower bound the $\ell_1$ distance as follows:
\begin{equation}
\mathcal{D}^{|t|}(\mbX, \mbY) \leq \mathcal{D}(\mbX, \mbY) \leq \mathcal{D}^{-|t|}(\mbX, \mbY) .
\end{equation}
This distance has some interesting properties which will be used in our theorems \cite[Lemma~1]{jacques2017small}:
\begin{equation}
\label{prop}
\begin{aligned}
\left|d^t(a, b)-d^s(a, b)\right| & \leq 4(\Delta+|t-s|), \\
\left|d^t(a, b)-|a-b|\right| & \leq 4(\Delta+|t|) .
\end{aligned}
\end{equation}
Another crucial tool for our proof sketches of the proposed theorems is the Hoeffding concentration inequality, defined as follows:
\begin{lemma}\cite[Theroem~2.6.2]{vershynin2018high}
\label{hoff}
Let $X_1,\cdots,X_N$ be independent, mean zero, subgaussian random variables. Then, for every $t\geq 0$, we have
\begin{equation}
\label{z_4}
\mathbb{P}\left(\left|\sum_{i=1}^{N}X_i\right|\geq t\right)\leq2e^{-\frac{ct^2}{\sum_{i=1}^{N}\|X_i\|_{\psi_2}^2}},
\end{equation}
where $c$ is a positive constant.
\end{lemma}
In the following proposition, we will show that the quantized mapping $\mathcal{Q}_{\Delta, \mbT}\left(\mathcal{P}_{\Omega}\left(\mbX\right)\right)$ is a quasi-isometric embedding between $\left(\mathcal{K}^r\subset \mathbb{R}^{n_1\times n_2},\ell_1\right)$ and $\left(\mathcal{Q}\left(\mathcal{P}_{\Omega}\left(\mathcal{K}^r\right)\right)\subset \mathcal{A}^{n_1\times n_2}_{K}, \ell_1\right)$. Then, based on this embedding, an upper-bound for the recovery performance will be derived:
\begin{proposition}
\label{Theorem_q}
Define the set $\mathcal{K}^r$ as
\begin{equation}
\label{n_02}
\mathcal{K}^r=\left\{\mbX^{\prime}\in\mathbb{R}^{n_1\times n_2}\mid\operatorname{rank}(\mbX^{\prime})\leq r, \left\|\mbX^{\prime}\right\|_{\mathrm{max}}\leq\alpha\right\}\subset\mathbb{R}^{n_1\times n_2}.
\end{equation}
Consider a matrix $\mbX\in\mathcal{K}^r$. Now, assume that $m$ entries of $\mbX$, randomly selected with uniform sampling, undergo scalar quantization with dither values following $\tau_{i,j}\sim\mathcal{U}_{[-\frac{\Delta}{2},\frac{\Delta}{2}]}$, and a resolution of $\Delta$. With constants $c$, $c^{\prime}>0$, and $\varepsilon \in (0,1)$, it can be asserted that the following quasi-isometric quantized embedding holds with a probability of at least $1-2e^{-c^{\prime}\varepsilon^2 m}$ as
\begin{equation}
\label{a_50}
\left|\mathcal{D}\left(\mbX,\mbY\right)-\frac{1}{n_1 n_2}\left\|\mbX-\mbY\right\|_1\right|\leq \varepsilon \left\|\mbX-\mbY\right\|_{\mathrm{F}}+c \varepsilon\Delta.
\end{equation}
for all $\mbX, \mbY \in\mathcal{K}^r$ when the required number of samples must satisfy
\begin{equation}
m\gtrsim \varepsilon^{-2} r\left(n_1+n_2\right)\log\left(1+\frac{\|\mathcal{K}^r\|}{\Delta \varepsilon}\right).   
\end{equation}
\end{proposition}
The proof of Proposition~\ref{Theorem_q} is provided in Appendix~\ref{ape1}. In the derived quasi-isometry, certain differences arise compared to the quasi-isometry in \eqref{jaack}, primarily due to the effect of uniform sampling in matrix completion. In this setting, we can directly obtain the expected value of $\mathcal{D}$, which appears in the form of an $\ell_1$ norm, without the need to bound it by the expected Gaussian norm $\left(\frac{2}{\pi}\right)^{\frac{1}{2}}\|\cdot\|_2$ to achieve a controllable or meaningful quantity. As shown in Appendix~\ref{ape1}, the Frobenius norm term can be replaced by a factor $2\alpha\varepsilon$, or even expressed in terms of the $\ell_1$ norm:
\begin{equation}
\frac{1}{n_1 n_2}\left\|\mbX-\mbY\right\|_1-c_1 \varepsilon(\alpha+\Delta)\leq\mathcal{D}\left(\mbX,\mbY\right)\leq \frac{1}{n_1 n_2}\left\|\mbX-\mbY\right\|_1+c_1 \varepsilon(\alpha+\Delta).    
\end{equation} 
An important observation is that, even in this form of quasi-isometry, the upper bound decreases as the number of samples increases.
\begin{theorem}
\label{recovery}
By assuming the consistent property from Definition~\ref{def_1}, and the quasi-isometric embedding provided in Proposition~\ref{Theorem_q}, the recovery error of the quantized matrix completion for $(\mbX,\mbY)\in\mathcal{K}^r$ with can be upper-bounded with a probability of at least $1-2e^{-c^{\prime}\varepsilon^2 m}$ as
\begin{equation}
\left\|\mbX-\mbY\right\|_1\lesssim\varepsilon n_1 n_2\left(\alpha+\Delta\right).    
\end{equation}
\end{theorem}
The proof of Theorem~\ref{recovery} is investigated in Appendix~\ref{ape2}. It is worth noting that if the dynamic range of measurements satisfies \eqref{a7}, i.e., $\alpha<\frac{\Delta}{2}$, we can find the upper-bound for the one-bit matrix completion based on the derived quasi-isometry as follows:
\begin{equation}
\left\|\mbX-\mbY\right\|_1\lesssim \varepsilon n_1 n_2\Delta.  
\end{equation}
In the following corollary, we establish an upper bound on the rate of decay with respect to the number of samples $m$ for any consistent solver addressing the dithered quantized matrix completion problem.
\begin{corollary}
\label{col_1}
For any consistent solver addressing the dithered quantized matrix completion problem, the rate of decay with respect to the number of samples $m$ is at most of order $\mathcal{O}\left(m^{-\frac{1}{2}}\right)$.
\end{corollary}
\begin{IEEEproof}
From Theorem~\ref{recovery}, the lower bound on the number of samples $m$ depends on $\varepsilon^{-2}$. Consequently, the decay rate of $\varepsilon$ with respect to $m$ is $\mathcal{O}\left(m^{-\frac{1}{2}}\right)$. Therefore, the rate of decay for any consistent solver is at most $\mathcal{O}\left(m^{-\frac{1}{2}}\right)$, completing the proof.
\end{IEEEproof}

\section{Estimates Analysis of One-Bit Matrix Completion}
\label{pro}
As established in \eqref{a5} and \eqref{a7}, one-bit quantization, representing an extreme instance of quantization, can be viewed as a special case of scalar quantization when certain conditions on the resolution parameter $\Delta$ are satisfied. The theoretical guarantees presented in Section~\ref{OB-MC} generally apply to scalar quantization, including the limiting one-bit case. However, these guarantees were derived using a discontinuous mapping $\mathcal{D}(\mbX,\mbY)$ and by subsequently relating this mapping to the pseudo-distance $\mathcal{D}^t(\mbX,\mbY)$. This analysis led to an $\ell_1$-norm error recovery guarantee that holds with high probability. In contrast, the analysis in this section focuses specifically on the one-bit quantization problem by constructing a Lipschitz-continuous mapping, which both simplifies the derivations and enables a Frobenius-norm error recovery guarantee.

We begin by formulating the dithered one-bit quantization process as a \emph{linear feasibility system} and subsequently reformulate the dithered one-bit matrix completion problem as a nuclear norm minimization task. We then introduce the concept of the \emph{Finite-Volume Property} and, building upon this property, construct a continuous mapping that forms the basis for establishing our Frobenius-norm error recovery guarantee.

\subsection{Dithered One-Bit Quantization}
\label{ss_4}
In one-bit quantization, introducing a dithering sequence significantly improves reconstruction performance compared to the ditherless setting \cite{eamaz2022covariance}. The dithering sequence can, in general, be drawn from an arbitrary distribution. In this work, however, our theoretical guarantees are established under the assumption that the dithering sequence is uniformly distributed. Accordingly, for one-bit quantization with such dithering, each one-bit measurement $r_k$ is obtained as $r_{k} = \operatorname{sgn}\left(x_{k}-\tau_{k}\right)$, where $\tau_k$ denotes the corresponding dither value.

The information obtained through one-bit sampling with a dithering sequence can be expressed as a system of linear inequalities. Specifically, each one-bit measurement satisfies $r_{k}=+1$ when $x_{k}\geq\tau_{k}$ and $r_{k}=-1$ when $x_{k}<\tau_{k}$. By stacking the signal entries into $\mathbf{x}=[x_{k}] \in \mathbb{R}^{n}$ and the corresponding one-bit measurements into $\mathbf{r}_x=[r_{k}] \in \{-1,1\}^{n}$, the feasible region that characterizes the geometric location of $\mathbf{x}$ can be expressed as
\begin{equation}
\label{eq:4}
r_{k}\left(x_{k}-\tau_{k}\right) \geq 0.
\end{equation}
The vectorized form of \eqref{eq:4} can be written as $\mathbf{r}_x \odot \left(\mathbf{x}-\btau\right) \succeq \mathbf{0}$, or equivalently,
\begin{equation}
\label{eq:6}
\begin{aligned}
\bOmega_{x} \mathbf{x} &\succeq \mathbf{r}_x \odot \btau,
\end{aligned}
\end{equation}
where $\bOmega_{x} \triangleq \operatorname{diag}\left(\mathbf{r}_x\right)$. This linear system of inequalities derived from the one-bit sampling scheme can be reformulated as a \textit{one-bit polyhedron} defined by
\begin{equation}
\label{eq:8n}
\begin{aligned}
\mathcal{P}_{x}=\left\{\mathbf{x}^{\prime}\in\mathbb{R}^n \mid \bOmega_{x} \mathbf{x}^{\prime} \succeq \mathbf{r}_x \odot \btau\right\}\subset \mathbb{R}^n.
\end{aligned}
\end{equation}
In the subsequent section, we make use of the one-bit polyhedron in \eqref{eq:8n} to characterize the feasible region of the one-bit matrix completion problem.

\subsection{Dithered One-Bit Matrix Completion}
\label{ss_5}
In the one-bit matrix completion problem, we only observe the one-bit data matrix $\mbR_x=\left[r_{i,j}\right]\in\{-1,0,1\}^{n_1\times n_2}$, whose entries are determined by comparing the corresponding elements of the sampled matrix $\mathcal{P}_{\Omega}\left(\mbX\right)$ with those of a dithering matrix $\mbT=\left[\tau_{i,j}\right]\in \mathbb{R}^{n_1\times n_2}$ according to the following relationship:
\begin{equation}
\label{St_2}
\begin{aligned}
r_{i,j} &= \begin{cases} +1 & X_{i,j}\geq\tau_{i,j},\\ -1 & X_{i,j}<\tau_{i,j},
\end{cases} \quad (i,j)\in\Omega,
\end{aligned}
\end{equation}
and $r_{i,j}=0$, for all $(i,j)\notin\Omega$.
Let $\mbP\in\{0,1\}^{m\times n_1 n_2}$ denote a permutation matrix that selects only the observed entries indexed by $\Omega$. Using this notation, the one-bit measurement model can be expressed as the following linear feasibility system:
\begin{equation}
\label{St_3}
\mbP\bOmega_x\operatorname{vec}\left(\mbX\right) \succeq\mbP\left(\operatorname{vec}\left(\mbR_x\right)\odot\operatorname{vec}\left(\mbT\right)\right),
\end{equation}
where $\bOmega_x=\operatorname{diag}\left(\operatorname{vec}\left(\mbR_x\right)\right)$.
The corresponding feasible set for recovering the low-rank matrix $\mbX$ from its one-bit observations is thus characterized as
\begin{equation}
\label{St_6}
\begin{aligned}
\mathcal{F}_{\mbX}= \left\{\mbX^{\prime}\in\mathbb{R}^{n_1\times n_2} \mid \mbP\bOmega_x\operatorname{vec}\left(\mbX^{\prime}\right) \succeq\mbP\left(\operatorname{vec}\left(\mbR_x\right)\odot\operatorname{vec}\left(\mbT\right)\right),~\left\|\mbX^{\prime}\right\|_{\star}\leq\epsilon\right\}\subset \mathbb{R}^{n_1\times n_2},
\end{aligned}
\end{equation}
where $\epsilon$ is a predefined threshold. To estimate $\mbX$, the one-bit matrix completion problem is formulated as a \emph{nuclear norm minimization} task:
\begin{equation}
\label{St_7}
\begin{aligned}
\mathcal{P}_{\mbX}:\quad\underset{\mbX^{\prime}}{\textrm{minimize}} \quad &\tau\left\|\mbX^{\prime}\right\|_{\star}+\frac{1}{2}\left\|\mbX^{\prime}\right\|^{2}_{\mathrm{F}}\\ \text{subject to} \quad &\mbP\bOmega_x\operatorname{vec}\left(\mbX^{\prime}\right) \succeq\mbP\left(\operatorname{vec}\left(\mbR_x\right)\odot\operatorname{vec}\left(\mbT\right)\right),
\end{aligned}
\end{equation}
for some fixed $\tau\geq 0$.
More than nuclear norm, the Frobenius norm is also considered to control the amplitudes of the unknown data \cite{cai2010singular}.

\subsection{Error Recovery Guarantee}
\label{er_one}
In contrast to the discontinuous mapping used in Section~\ref{OB-MC} for quantized matrix completion, the analysis in this section for one-bit quantization is built upon a continuous mapping. To elaborate, let the matrix $\mbX=[X_{i,j}]\in\mathbb{R}^{n_1\times n_2}$ denote the true matrix that we aim to recover through the program $\mathcal{P}_\mbX$ defined in \eqref{St_7}. Let $d_{i,j}$ represent the distance between $X_{i,j}$ and its corresponding constraint hyperplanes in $\mathcal{P}_\mbX$. A straightforward derivation shows that this distance takes the form $d_{i,j}=\left|X_{i,j}-\tau_{i,j}\right|$ for all $(i,j)\in\Omega$. As the number of observed samples increases, i.e., as $m$ grows, additional hyperplanes are introduced into the constraint set of $\mathcal{P}_\mbX$. The intersection of these hyperplanes defines the feasible region, and any point within this region that also satisfies the low-rank constraint can serve as a valid solution to $\mathcal{P}_\mbX$. Interestingly, as the number of samples $m$ increases, the feasible region becomes progressively smaller, thereby improving the likelihood that the solution to $\mathcal{P}_\mbX$ lies close to the true matrix $\mbX$. Furthermore, as $m$ grows, the empirical mean of the distances $\{d_{i,j}\}_{(i,j)\in\Omega}$ converges to its expected value. Consequently, we argue that as the number of samples increases, any solution that satisfies the hyperplane constraints of $\mathcal{P}_\mbX$ will exhibit a smaller error with respect to the true matrix $\mbX$ with high probability. We refer to this phenomenon as the ``Finite-Volume Property''. To formally establish our result, we define the following operator:
\begin{definition}
\label{def_T}
For a matrix $\mbX=[X_{i,j}]\in\mathbb{R}^{n_1\times n_2}$ and a dithering matrix $\mbT=\left[\tau_{i,j}\right]\in \mathbb{R}^{n_1\times n_2}$, define $d_{i,j}=\left|X_{i,j}-\tau_{i,j}\right|$ as a distance between the $(i,j)$-th entries of $\mbX$ and $\mbT$ for all $(i,j)\in\Omega$. We then define the empirical average of these distances as
\begin{equation}
\label{a_3}
\begin{aligned}
T(\mbX)=\frac{1}{m}\sum_{(i,j)\in\Omega}\left|X_{i,j}-\tau_{i,j}\right|,
\end{aligned}
\end{equation}
where $|\Omega|=m$.
\end{definition}
Before presenting our main result on one-bit matrix completion, we first introduce the notion of consistency in the context of dithered one-bit quantization, which we formally define as follows:
\begin{definition}
\label{def_2}
Define a low-rank matrix as $\mbX=[X_{i,j}]\in\mathbb{R}^{n_1\times n_2}$ and let the dither matrix be $\mbT=\left[\tau_{i,j}\right]\in \mathbb{R}^{n_1\times n_2}$. Let $\mbY=[Y_{i,j}]\in\mathbb{R}^{n_1\times n_2}$ denote the estimate produced by an arbitrary reconstruction algorithm addressing problem \eqref{St_7}. We say that such a reconstruction algorithm is consistent if
\begin{equation}
\label{con1}
\operatorname{sgn}\left(X_{i,j}-\tau_{i,j}\right)=\operatorname{sgn}\left(Y_{i,j}-\tau_{i,j}\right),~(i,j)\in\Omega,
\end{equation}
or in the matrix form
\begin{equation}
\label{con2}
\operatorname{sgn}\left(\mathcal{P}_{\Omega}\left(\mbX-\mbT\right)\right)=\operatorname{sgn}\left(\mathcal{P}_{\Omega}\left(\mbY-\mbT\right)\right).
\end{equation}
\end{definition}
The notion of consistency in Definition~\ref{def_2} is analogous to that in Definition~\ref{def_1}. In this case, instead of satisfying a multi-bit quantization cell, both $\mbX$ and $\mbY$ are required to satisfy the same set of linear inequalities specified in the program $\mathcal{P}_\mbX$ in \eqref{St_7}.
In the following theorem, we establish a universal one-bit matrix completion guarantee that holds for any consistent reconstruction algorithm as defined in Definition~\ref{def_2}.
\begin{theorem}
\label{thr_onebit}
Consider the set 
\begin{equation}
\label{n_2}
\mathcal{K}^r=\left\{\mbX^\prime\in\mathbb{R}^{n_1\times n_2}\mid\operatorname{rank}\left(\mbX^\prime\right)\leq r,\left\|\mbX^\prime\right\|_{\mathrm{max}}\leq\alpha\right\}.
\end{equation}
Suppose $m$ entries of $\mbX$ with locations sampled uniformly at random are compared with a sequence of uniform dithers generated as $\tau_{i,j}\sim \mathcal{U}_{\left[-\zeta,\zeta\right]}$ for all $(i,j)\in\Omega$ and $\zeta\geq\alpha$, resulting in the observed one-bit data. With a universal constant $c>0$ and $\epsilon\in(0,1)$, the following recovery bound holds with probability at least $1-3e^{-c\epsilon^2m}$ for all matrices $\mbX,\mbY\in\mathcal{K}^r$ that satisfy the consistent reconstruction property in Definition~\ref{def_2}:
\begin{equation}
\label{a_5}
\|\mbX-\mbY\|_{\mathrm{F}}\leq 4\sqrt{\epsilon\zeta n_1n_2}.
\end{equation}
This guarantee holds provided that the number of samples satisfies
\begin{equation}
\label{samplereq}
m\gtrsim\epsilon^{-2}r(n_1+n_2)\log\left(1+\epsilon^{-1}+\|\mathcal{K}^r\|\right).
\end{equation}
\end{theorem}
The proof of Theorem~\ref{thr_onebit} is presented in Appendix~\ref{ape3}. In the following corollary, we establish an upper bound on the rate of decay with respect to the number of samples $m$ for any consistent solver addressing the dithered one-bit matrix completion problem.
\begin{corollary}
\label{col_2}
For any consistent solver addressing the dithered one-bit matrix completion problem, the rate of decay with respect to the number of samples $m$ is at most of order $\mathcal{O}\left(m^{-\frac{1}{4}}\right)$.
\end{corollary}
The proof of Corollary~\ref{col_2} is identical to that of Corollary~\ref{col_1} and is omitted here.
As stated in Theorem~\ref{thr_onebit}, the result holds for any consistent solver; i.e., for any approximate solution $\mbY = [Y_{i,j}]$ satisfying \eqref{con1}. In the following theorem, we extend this result to a broader class of solvers that may not yield consistent approximate solutions as defined in Definition~\ref{def_2}.
\begin{theorem}
\label{thr_2}
Under the assumptions of Theorem~\ref{thr_onebit}, with probability at least $1 - 3e^{-c\epsilon^2 m}$ for some universal constant $c > 0$ and $\epsilon \in (0,1)$, the following upper recovery bound holds for all $\mbX, \mbY \in \mathcal{K}^r$ that do not satisfy the consistency property in Definition~\ref{def_2}:
\begin{equation}
\label{noc}
\|\mbX-\mbY\|_{\mathrm{F}}\leq 4\sqrt{\epsilon\zeta n_1n_2+\zeta^2n_1n_2d_{\mathrm{H}}\left(\operatorname{vec}\left(\mbR_x\right),\operatorname{vec}\left(\mbR_y\right)\right)},
\end{equation}
where $\mbR_x$ and $\mbR_y$ denote the one-bit measurements corresponding to $\mathcal{P}_{\Omega}\left(\mbX\right)$ and $\mathcal{P}_{\Omega}\left(\mbY\right)$, respectively.
\end{theorem}
The proof of Theorem~\ref{thr_2} is provided in Appendix~\ref{ape4}. In general, the distance between $\mbX$ and $\mbY$ is expected to increase with high probability when they do not belong to the same consistency cell. The result of Theorem~\ref{thr_2} formalizes this phenomenon through the presence of the term $d_{\mathrm{H}}\left(\operatorname{vec}\left(\mbR_x\right),\operatorname{vec}\left(\mbR_y\right)\right)$ in the upper recovery bound. Notably, as $d_{\mathrm{H}}\left(\operatorname{vec}\left(\mbR_x\right),\operatorname{vec}\left(\mbR_y\right)\right)\rightarrow 0$; i.e., when $\mbX$ and $\mbY$ lie within the same consistency cell, the bound in Theorem~\ref{thr_2} naturally reduces to that of Theorem~\ref{thr_onebit}.


\appendices
\section{Proof of Proposition~\ref{Theorem_q}}
\label{ape1}
Let us define $\Tilde{d}_{i,j} = \left|\left[\mathcal{P}_{\Omega}\left(\mbX-\mbY\right)\right]_{i,j}\right|$ with $(i,j)\in\Omega$, and $d^{t}_{i,j}=d^t\left(X_{i,j}+\tau_{i,j}, Y_{i,j}+\tau_{i,j}\right)$, where $d^t(\cdot)$ is defined in \eqref{e1}. For a subgaussian random vector $\bpsi$,
we have $\left\|\langle\bpsi,\mbx-\mby\rangle\right\|_{\psi_2}\lesssim \|\mbx-\mby\|_2$. In our setting, since the sensing matrix corresponds to uniform sampling via a permutation matrix, it follows directly that
\begin{equation}
\left\|\Tilde{d}_{i,j}\right\|_{\psi_2}=  \left\|\left[\mathcal{P}_{\Omega}\left(\mbX-\mbY\right)\right]_{i,j}\right\|_{\psi_2}\lesssim \left\|\mbX-\mbY\right\|_{\mathrm{F}},~(i,j)\in\Omega.
\end{equation}
If one instead wishes to express the bound in terms of the infinite norm, we can write
\begin{equation}
\label{37}
\left\|\left[\mathcal{P}_{\Omega}\left(\mbX-\mbY\right)\right]_{i,j}\right\|_{\psi_2}\lesssim \left\|\left[\mathcal{P}_{\Omega}\left(\mbX-\mbY\right)\right]_{i,j}\right\|_{\infty}\leq\left\|\mbX-\mbY\right\|_{\mathrm{max}} \leq 2\alpha,~(i,j)\in\Omega.
\end{equation}
This bound follows from the fact that, for bounded random variables, the subgaussian norm satisfies $\|\cdot\|_{\psi_2}\lesssim \|\cdot\|_{\infty}$.

Thus, for the subgaussian norm of $d^{t}_{i,j}$, 
by invoking the result in \eqref{prop} and using the bound $\|\cdot\|_{\psi_2}\lesssim \|\cdot\|_{\infty}$, we can write
\begin{equation}
\begin{aligned}
\label{subgaussbound}
\left\|d^t_{i,j}\right\|_{\psi_2} & \leq\left\|d^t_{i,j}-\Tilde{d}_{i,j}\right\|_{\psi_2}+\|\Tilde{d}_{i,j}\|_{\psi_2} \\
& \lesssim\left\|d^t\left(X_{i,j}+\tau_{i,j}, Y_{i,j}+\tau_{i,j}\right)-|X_{i,j}-Y_{i,j}|\right\|_{\psi_2}+\left\|\mbX-\mbY\right\|_{\mathrm{F}} \\
& \lesssim \Delta+|t|+\left\|\mbX-\mbY\right\|_{\mathrm{F}},~(i,j)\in\Omega.
\end{aligned}
\end{equation}
For fixed matrices $\mbX,\mbY\in\mathcal{K}^r$, by applying Hoeffding's inequality together with \eqref{subgaussbound}, we obtain the following concentration bound:
\begin{equation}
\label{fixed}
\mathbb{P}\left(\left|\mathcal{D}^{t}\left(\mbX,\mbY\right)-\frac{1}{m}\sum_{(i,j)\in\Omega}\mathbb{E}d_{i,j}^t\right|\geq\varepsilon\left(\Delta+|t|+\left\|\mbX-\mbY\right\|_{\mathrm{F}}\right)\right)\leq2e^{-c \varepsilon^2 m},
\end{equation}
for some universal constant $c>0$. As discussed in Section~\ref{ss_3}, each term in the $\ell_1$ distance defined in \eqref{n1} corresponds to $d^0\left(X_{i,j}+\tau_{i,j}, Y_{i,j}+\tau_{i,j}\right)$ for $(i,j)\in\Omega$. We denote this random variable by $d_{i,j}=d^0\left(X_{i,j}+\tau_{i,j}, Y_{i,j}+\tau_{i,j}\right)$ for all $(i,j)\in\Omega$. To connect the expected pseudo-distance $d^t_{i,j}$ with the expected $\ell_1$
distance characterized in Lemma~\ref{yuta}, we establish the following bound:
\begin{equation}
\begin{aligned}
\left|\mathbb{E} d^t_{i,j}-\mathbb{E} d_{i,j}\right|\leq \mathbb{E}\left|d^t_{i,j}-d_{i,j}\right|=\mathbb{E}_{(i,j)} \mathbb{E}_{\tau}\left|d^t\left(X_{i,j}+\tau, Y_{i,j}+\tau\right)-d^0\left(X_{i,j}+\tau, Y_{i,j}+\tau\right)\right|,
\end{aligned}
\end{equation}
where, for notational simplicity, the random variable $\tau_{i,j}$ is denoted by $\tau$. By taking the expectation over $\tau$, we obtain the following bound:
\begin{equation}
\begin{aligned}
& \mathbb{E}_{\tau}\left|d^t\left(X_{i,j}+\tau, Y_{i,j}+\tau\right)-d^0\left(X_{i,j}+\tau, Y_{i,j}+\tau\right)\right|\\
& \leq \Delta \sum_{k \in \mathbb{Z}} \mathbb{E}_{\tau} \mathbb{I}\left[\left\{\left|X_{i,j}+\tau-k \Delta\right| \leq|t|\right\} \cup\left\{\left|Y_{i,j}+\tau-k \Delta\right| \leq|t|\right\}\right] \\
& \leq \Delta \sum_{k \in \mathbb{Z}} \mathbb{E}_{\tau} \mathbb{I}\left[\left\{\left|X_{i,j}+\tau-k \Delta\right| \leq|t|\right\}\right]+\Delta \sum_{k \in \mathbb{Z}} \mathbb{E}_{\tau} \mathbb{I}\left[\left\{\left|Y_{i,j}+\tau-k \Delta\right| \leq|t|\right\}\right].
\end{aligned}
\end{equation}
For $\tau\sim\mathcal{U}_{[-\frac{\Delta}{2},\frac{\Delta}{2}]}$
\begin{equation}
\begin{aligned}
\Delta \sum_{k \in \mathbb{Z}} \mathbb{E}_{\tau} \mathbb{I}\left[\left\{\left|X_{i,j}+\tau-k \Delta\right| \leq|t|\right\}\right]
&=\sum_{k \in \mathbb{Z}} \int_{-\frac{\Delta}{2}}^{\frac{\Delta}{2}} \mathbb{I}\left[\left\{\left|X_{i,j}+\tau-k \Delta\right| \leq|t|\right\}\right] \mathrm{d} \tau \\
&=\int_{\mathbb{R}} \mathbb{I}\left[\left\{\left|X_{i,j}+\tau\right| \leq|t|\right\}\right] \mathrm{d} \tau=2|t|.
\end{aligned}
\end{equation}
Analogous results hold for $Y_{i,j}$, i.e., $ \Delta \sum_{k \in \mathbb{Z}} \mathbb{E}_{\tau} \mathbb{I}\left[\left\{\left|Y_{i,j}+\tau-k \Delta\right| \leq|t|\right\}\right]=2|t|$. Consequently, since both quantities are independent of the randomness of $(i,j)$, we obtain
\begin{equation}
\left|\mathbb{E} d^t_{i,j}-\mathbb{E} d_{i,j}\right| \lesssim  |t|.
\end{equation}
To relate $\mathcal{D}^{t}\left(\mbX,\mbY\right)$ to the expected $\ell_1$
distance characterized in Lemma~\ref{yuta}, we can write
\begin{equation}
\begin{aligned}
\left|\mathcal{D}^{t}\left(\mbX,\mbY\right)-\frac{1}{m}\sum_{(i,j)\in\Omega}\mathbb{E}d_{i,j}^t\right|&=\left|\mathcal{D}^{t}\left(\mbX,\mbY\right)-\frac{1}{m}\sum_{(i,j)\in\Omega}\mathbb{E}d_{i,j}^0-\frac{1}{m}\sum_{(i,j)\in\Omega}\left(\mathbb{E}d_{i,j}^t-\mathbb{E}d_{i,j}^0\right)\right|\\&\geq\left|\mathcal{D}^{t}\left(\mbX,\mbY\right)-\frac{1}{m}\sum_{(i,j)\in\Omega}\mathbb{E}d_{i,j}^0\right|-\frac{1}{m}\left|\sum_{(i,j)\in\Omega}\left(\mathbb{E}d_{i,j}^t-\mathbb{E}d_{i,j}^0\right)\right|\\&\geq\left|\mathcal{D}^{t}\left(\mbX,\mbY\right)-\frac{1}{n_1 n_2}\left\|\mbX-\mbY\right\|_1\right|-\frac{1}{m}\sum_{(i,j)\in\Omega}\left|\mathbb{E}d_{i,j}^t-\mathbb{E}d_{i,j}^0\right|\\&\geq\left|\mathcal{D}^{t}\left(\mbX,\mbY\right)-\frac{1}{n_1 n_2}\left\|\mbX-\mbY\right\|_1\right|-c_1|t|,
\end{aligned}
\end{equation}
for some constant $c_1>0$. Combining this result with \eqref{fixed}, we obtain
\begin{equation}
\mathbb{P}\left(\left|\mathcal{D}^{t}\left(\mbX,\mbY\right)-\frac{1}{n_1 n_2}\left\|\mbX-\mbY\right\|_1\right|\geq c_1|t|+\varepsilon\left(\Delta+|t|+\left\|\mbX-\mbY\right\|_{\mathrm{F}}\right)\right)\leq2e^{-c \varepsilon^2 m}.
\end{equation}
We now aim to extend this result to hold for all $\mbX,\mbY\in\mathcal{K}^r$.
Specifically, for all $(\mbX^{\prime}, \mbY^{\prime}) \in \mathcal{K}^r_{\rho}$, where for any $\mbA \in \mathcal{K}^{r}$ there exists $\mbA^{\prime} \in \mathcal{K}^{r}_{\rho}$ such that $\left\|\mbA - \mbA^{\prime}\right\|_{\mathrm{F}} \leq \rho$, the following concentration result holds:
\begin{equation}
\label{43}
\mathbb{P}\left(\sup_{\mbX^{\prime},\mbY^{\prime}\in\mathcal{K}^r_{\rho}}\left|\mathcal{D}^{t}\left(\mbX^{\prime},\mbY^{\prime}\right)-\frac{1}{n_1 n_2}\left\|\mbX^{\prime}-\mbY^{\prime}\right\|_1\right|\geq c_1|t|+\varepsilon\left(\Delta+|t|+\left\|\mbX^{\prime}-\mbY^{\prime}\right\|_{\mathrm{F}}\right)\right)\leq2e^{2\mathcal{H}\left(\mathcal{K}^r,\rho\right)-c \varepsilon^2 m},
\end{equation}
For structured sets such as the low-rank matrix set, the upper bound on the Kolmogorov $\rho$-entropy is given by \cite[Table~1]{jacques2017time}:
\begin{equation}
\label{kol_bound}
\mathcal{H}\left(\mathcal{K}^r,\rho\right)\lesssim r \left(n_1+n_2\right)\log\left(1+\frac{\|\mathcal{K}^r\|}{\rho}\right).
\end{equation}
Using this result, we can readily verify that the concentration bound in \eqref{43} holds with failure probability at most $2e^{-c^{\prime}\varepsilon^2m}$ for some universal constant $c^\prime>0$, provided that the number of samples $m$ satisfies
\begin{equation}
\label{sample_m}
m\gtrsim \varepsilon^{-2} r\left(n_1+n_2\right)\log\left(1+\frac{\|\mathcal{K}^r\|}{\rho}\right).   
\end{equation}
The significance of the pseudo-distance $\mathcal{D}^t$ becomes evident here as we aim to extend these results to all pairs of matrices in $\mathcal{K}^{r}$ by examining the continuity properties of $\mathcal{D}^t$ within a limited neighborhood around the selected matrices. To initiate this analysis, we first establish a bound on the measurements (as detailed below), which quantifies the error between elements of $\mathcal{K}^r$ and $\mathcal{K}^r_\rho$ in the context of matrix completion.

Let us define $\widehat{\mbX}=\mbX-\mbX^{\prime}$, and $\widehat{\mbY}=\mbY-\mbY^{\prime}$ where $\left(\widehat{\mbX},\widehat{\mbY}\right)\in\left(\mathcal{K}^{r}-\mathcal{K}^r\right)\bigcap \rho \mathbb{B}_{\mathrm{F}}^{n_1\times n_2}$. It then follows that
\begin{equation}
\label{32}
\left\|\mathcal{P}_{\Omega}\left(\widehat{\mbX}\right)\right\|_{\mathrm{F}}\leq \left\|\widehat{\mbX}\right\|_{\mathrm{F}}\leq \rho.
\end{equation}
We now analyze the continuity of the pseudo-distance with respect to Frobenius-norm perturbations of matrices. Our approach follows that of \cite{jacques2017small}, with a key distinction: unlike the subgaussian measurement bound of $\rho\sqrt{m}$, the matrix completion measurements in our setting are bounded by $\rho$, as shown in \eqref{32}. This difference necessitates a corresponding adjustment of the parameters in \cite[Lemma~3]{jacques2017small}. The following lemma is instrumental in establishing that the pseudo-distance between matrices in $\mathcal{K}^r$ can be controlled via their perturbed counterparts in $\mathcal{K}^r_{\rho}$, which is essential for extending the result in \eqref{43} to all matrix pairs in the main set $\mathcal{K}^r$.
\begin{lemma}
\label{l22}
Assume that $\left\|\mathcal{P}_{\Omega}\left(\widehat{\mbX}\right)\right\|_{\mathrm{F}}\leq \rho$, and $\left\|\mathcal{P}_{\Omega}\left(\widehat{\mbY}\right)\right\|_{\mathrm{F}}\leq \rho$. Then for every $t\in\mathbb{R}$ and $0<P\leq 1$ we have
\begin{equation}
\begin{aligned}
\mathcal{D}^{t+\rho \sqrt{P}}\left(\mbX^{\prime}, \mbY^{\prime}\right)- \frac{4}{m}\left(\frac{2\Delta}{P}+\frac{\rho}{\sqrt{P}}\right) \leq \mathcal{D}^t\left(\mbX, \mbY\right)  \leq \mathcal{D}^{t-\rho \sqrt{P}}\left(\mbX^{\prime}, \mbY^{\prime}\right)+\frac{4}{m}\left(\frac{2\Delta}{P}+\frac{\rho}{\sqrt{P}}\right).
\end{aligned}
\end{equation}
\end{lemma}
\begin{IEEEproof}
The proof follows the approach of \cite[Appendix~D]{jacques2017small}, differing only in the upper bounds applied to our measurements. In Lemma~\ref{l22}, the parameter $P$ serves as a smoothing factor that helps manage the discontinuities introduced by quantization in the pseudo-distance $D^t$. Quantization induces sudden jumps whenever a projected value crosses a quantization threshold; the analysis addresses these abrupt transitions by replacing them with a softened or averaged version controlled by $P$. This smoothing ensures that $D^t$ remains continuous with respect to small Frobenius-norm perturbations of its arguments, despite the inherently discrete nature of quantization.

To establish this result, the proof introduces two complementary index sets, $T$ and $T^c$, which separate the ``regular'' and ``irregular'' measurement components. The set $T$ is defined as
\[
T := \left\{\, (i,j) \in \Omega,~ \left(\widehat{\mbX},\widehat{\mbY}\right)\in\left(\mathcal{K}^{r}-\mathcal{K}^r\right)\cap \rho \mathbb{B}_{\mathrm{F}}^{n_1\times n_2} : \left|\widehat{X}_{i,j}\right| \leq \rho \sqrt{P}, \; \left| \widehat{Y}_{i,j} \right| \leq \rho \sqrt{P} \,\right\},
\]
i.e., the collection of indices for which the perturbations of $\mbX$ and $\mbY$ remain small enough that their projected values do \textit{not} cross a quantization boundary. In these components, the perturbed projections stay within the same quantization bin as the unperturbed ones, and the pseudo-distance behaves smoothly. The complement $T^c$ includes the remaining indices where the perturbations are large enough to cause at least one quantization threshold crossing, resulting in discontinuous changes in the quantized outputs. By analyzing the contributions of each set separately, the proof shows that most indices lie in $T$, where continuity holds directly, while the total contribution from $T^c$ can be bounded 
and controlled via the parameters $\Delta$, $\rho$, and $P$. This partitioning is crucial for establishing that the pseudo-distance $D^t$ changes only slightly under small Frobenius-norm perturbations, thereby ensuring the desired continuity property needed to extend local results to all pairs of matrices in the considered set. Moreover, the cardinality of $T^{c}$ appears explicitly in the proof and plays a significant role in determining the resulting bounds. It is evaluated as follows:
\begin{equation}
\left\|\mathcal{P}_{\Omega}\left(\widehat{\mbX}\right)\right\|^2_{\mathrm{F}}+\left\|\mathcal{P}_{\Omega}\left(\widehat{\mbY}\right)\right\|^2_{\mathrm{F}}\leq 2 \rho^2,
\end{equation}
and
\begin{equation}
|T^c|\rho^2 P+\left\|\mathcal{P}_{T(\Omega)}\left(\widehat{\mbX}\right)\right\|^2_{\mathrm{F}}+\left\|\mathcal{P}_{T(\Omega)}\left(\widehat{\mbY}\right)\right\|^2_{\mathrm{F}}\leq\left\|\mathcal{P}_{\Omega}\left(\widehat{\mbX}\right)\right\|^2_{\mathrm{F}}+\left\|\mathcal{P}_{\Omega}\left(\widehat{\mbY}\right)\right\|^2_{\mathrm{F}},
\end{equation}
where \( T(\Omega) \) refers to the indices of \( T \subseteq \Omega\). As a result, one can readily obtain
\begin{equation}
|T^c|\leq \frac{2}{P}.    
\end{equation}
Considering the definition of $\mathcal{F}^{t}$ in \eqref{def_f_t}, we have, for all $(i,j)\in T$ and any $\lambda\in\mathbb{R}$
\begin{equation}
\begin{aligned}
\label{ft_subset}
\mathcal{F}^{t+\rho\sqrt{P}}\left(X^{\prime}_{i,j}+\tau_{i,j}-\lambda,Y^{\prime}_{i,j}+\tau_{i,j}-\lambda\right)&\subset \mathcal{F}^{t}\left(X_{i,j}+\tau_{i,j}-\lambda,Y_{i,j}+\tau_{i,j}-\lambda\right)\\&\subset    
\mathcal{F}^{t-\rho\sqrt{P}}\left(X^{\prime}_{i,j}+\tau_{i,j}-\lambda,Y^{\prime}_{i,j}+\tau_{i,j}-\lambda\right).
\end{aligned}
\end{equation}
Denoting $A_{i,j}=\max\left\{\left|\widehat{X}_{i,j}\right|,\left|\widehat{Y}_{i,j}\right|\right\}$, we find
\begin{equation}
\begin{aligned}
\mathcal{D}^{t+\rho \sqrt{P}}\left(\mbX^{\prime}, \mbY^{\prime}\right) & =\frac{\Delta}{m} \sum_{(i,j)\in\Omega} \sum_{k \in \mathbb{Z}} \mathbb{I}\left[\mathcal{F}^{t+\rho \sqrt{P}}\left(X^{\prime}_{i,j}+\tau_{i,j}-k\Delta, Y^{\prime}_{i,j}+\tau_{i,j}-k\Delta\right)\right] \\
\leq & \frac{\Delta}{m} \sum_{(i,j) \in T} \sum_{k \in \mathbb{Z}} \mathbb{I}\left[\mathcal{F}^t\left(X_{i,j}+\tau_{i,j}-k\Delta, Y_{i,j}+\tau_{i,j}-k\Delta\right)\right] \\
& +\frac{\Delta}{m} \sum_{(i,j) \in T^c} \sum_{k \in \mathbb{Z}} \mathbb{I}\left[\mathcal{F}^{t+\rho \sqrt{P}-A_{i,j}}\left(X_{i,j}+\tau_{i,j}-k\Delta, Y_{i,j}+\tau_{i,j}-k\Delta\right)\right] \\
\leq & \frac{\Delta}{m} \sum_{(i,j) \in T} \sum_{k \in \mathbb{Z}} \mathbb{I}\left[\mathcal{F}^t\left(X_{i,j}+\tau_{i,j}-k\Delta, Y_{i,j}+\tau_{i,j}-k\Delta\right)\right] \\
& +\frac{\Delta}{m} \sum_{(i,j) \in T^c} \sum_{k \in \mathbb{Z}} \mathbb{I}\left[\mathcal{F}^t\left(X_{i,j}+\tau_{i,j}-k\Delta, Y_{i,j}+\tau_{i,j}-k\Delta\right)\right] \\
& 
+\frac{1}{m} \sum_{(i,j) \in T^c} \Delta \sum_{k \in \mathbb{Z}} \Biggl\rvert\, 
\mathbb{I}\left[\mathcal{F}^{t+\rho \sqrt{P}-A_{i,j}}\left(X_{i,j}+\tau_{i,j}-k\Delta, Y_{i,j}+\tau_{i,j}-k\Delta\right)\right] \\
& -\mathbb{I}\left[\mathcal{F}^t\left(X_{i,j}+\tau_{i,j}-k\Delta, Y_{i,j}+\tau_{i,j}-k\Delta\right)\right]\Biggr\rvert.
\end{aligned}
\end{equation}
By leveraging the pseudo-distance properties established in \eqref{prop}, we obtain
\begin{equation}
\begin{aligned}
\mathcal{D}^{t+\rho \sqrt{P}}\left(\mbX^{\prime}, \mbY^{\prime}\right)
& \leq \mathcal{D}^t\left(\mbX, \mbY\right)+\frac{4}{m} \sum_{(i,j) \in T^c}\left(\Delta+A_{i,j}-\rho \sqrt{P}\right) \\
& \leq\mathcal{D}^t\left(\mbX, \mbY\right)+\frac{8 \Delta}{m P}+\frac{4}{m} \sum_{(i,j) \in T^c}\left(A_{i,j}-\rho\sqrt{P}\right).
\end{aligned}
\end{equation}
On the other hand, we have
\begin{equation}
\begin{aligned}
\frac{1}{m} \sum_{(i,j) \in T^c}A_{i,j}&\leq \frac{1}{m}\left(\left\|\mathcal{P}_{T^c(\Omega)}\left(\widehat{\mbX}\right)\right\|_1+\left\|\mathcal{P}_{T^c(\Omega)}\left(\widehat{\mbY}\right)\right\|_1\right)\\
& \leq \frac{\sqrt{|T^c|}}{m}\left(\left\|\mathcal{P}_{T^c(\Omega)}\left(\widehat{\mbX}\right)\right\|_{\mathrm{F}}+\left\|\mathcal{P}_{T^c(\Omega)}\left(\widehat{\mbY}\right)\right\|_{\mathrm{F}}\right)\\
&\leq 2 \rho \frac{\sqrt{|T^c|}}{m},
\end{aligned}
\end{equation}
which leads to
\begin{equation}
\begin{aligned}
\mathcal{D}^{t+\rho \sqrt{P}}\left(\mbX^{\prime}, \mbY^{\prime}\right)
&\leq\mathcal{D}^t\left(\mbX, \mbY\right)+\frac{8 \Delta}{m P}+\frac{4\rho}{m}\left(2\sqrt{|T^c|}- |T^c| \sqrt{P}\right)\\
&\leq \mathcal{D}^t\left(\mbX, \mbY\right)+\frac{8 \Delta}{m P}+\frac{4\rho}{m\sqrt{P}}.
\end{aligned}
\end{equation}
The final inequality follows from the fact that, for any real number $t$, $2t-t^2\sqrt{P}\leq \frac{1}{\sqrt{P}}$. The lower bound can be derived in a similar manner by using the relationship established in \eqref{ft_subset}.
\end{IEEEproof}
Based on the concentration inequality in \eqref{43}, we have, with probability at least $1-2e^{c^{\prime}\varepsilon^2 m}$,
\begin{equation}
\begin{aligned}
\label{vv}
\left|\mathcal{D}^{(t-\rho\sqrt{P})}\left(\mbX^{\prime},\mbY^{\prime}\right)-\frac{1}{n_1n_2}\left\|\mbX^{\prime}-\mbY^{\prime}\right\|_1\right|\leq c_1\left|t-\rho\sqrt{P}\right|+\varepsilon\left(\Delta+\left|t-\rho\sqrt{P}\right|+\left\|\mbX^{\prime}-\mbY^{\prime}\right\|_{\mathrm{F}}\right),  
\end{aligned}
\end{equation}
and 
\begin{equation}
\label{ww}
\begin{aligned}
\left|\mathcal{D}^{(t+\rho\sqrt{P})}\left(\mbX^{\prime},\mbY^{\prime}\right)-\frac{1}{n_1n_2}\left\|\mbX^{\prime}-\mbY^{\prime}\right\|_1\right|\leq c_1\left|t+\rho\sqrt{P}\right|+\varepsilon\left(\Delta+\left|t+\rho\sqrt{P}\right|+\left\|\mbX^{\prime}-\mbY^{\prime}\right\|_{\mathrm{F}}\right),
\end{aligned}   
\end{equation}
provided that the number of samples $m$ satisfies \eqref{sample_m}.
Using Lemma~\ref{l22} and the upper bound in \eqref{vv}, we obtain
\begin{equation}
\begin{aligned}
\mathcal{D}^t(\mbX, \mbY) \leq & \mathcal{D}^{t-\rho \sqrt{P}}\left(\mbX^{\prime}, \mbY^{\prime}\right)+\frac{4}{m}\left(\frac{2\Delta}{P}+\frac{\rho}{\sqrt{P}}\right) \\
\leq & (c_1+\varepsilon)\left|t-\rho \sqrt{P}\right|+\frac{1}{n_1n_2}\left\|\mbX^{\prime}-\mbY^{\prime}\right\|_1 +\varepsilon\left\|\mbX^{\prime}-\mbY^{\prime}\right\|_{\mathrm{F}}+\varepsilon \Delta+\frac{4}{m}\left(\frac{2\Delta}{P}+\frac{\rho}{\sqrt{P}}\right).
\end{aligned}
\end{equation}
By the reverse triangle inequality, we have
\begin{equation}
\frac{1}{n_1n_2}\left|\left\|\mbX^{\prime}-\mbY^{\prime}\right\|_1-\left\|\mbX-\mbY\right\|_1  \right|\leq \frac{1}{n_1n_2}\left(\sqrt{n_1n_2}\left\|\widehat{\mbX}\right\|_{\mathrm{F}}+\sqrt{n_1n_2}\left\|\widehat{\mbY}\right\|_{\mathrm{F}}\right)\leq \frac{2\rho}{\sqrt{n_1n_2}}.
\end{equation}
Similarly, it follows that
\begin{equation}
\left|\left\|\mbX^{\prime}-\mbY^{\prime}\right\|_{\mathrm{F}}-\left\|\mbX-\mbY\right\|_{\mathrm{F}}  \right|\leq \left(\left\|\widehat{\mbX}\right\|_{\mathrm{F}}+\left\|\widehat{\mbY}\right\|_{\mathrm{F}}\right)\leq 2\rho.
\end{equation}
Assuming $\varepsilon < 1$, there exists a constant $c_2 > 0$ such that
\begin{equation}
\mathcal{D}^t\left(\mbX,\mbY\right)-\frac{1}{n_1n_2}\left\|\mbX-\mbY\right\|_1\leq \varepsilon \left\|\mbX-\mbY\right\|_{\mathrm{F}}+c_2\left(|t|+\rho\sqrt{P}+\rho+\varepsilon\Delta+\frac{\Delta}{mP}+\frac{\rho}{m\sqrt{P}}\right).
\end{equation}
Selecting the parameters such that $m P = \varepsilon^{-1}>1$ and $\rho=\Delta \varepsilon$, which gives $\frac{\rho}{m \sqrt{P}}=\frac{\varepsilon^{\frac{3}{2}}\Delta}{\sqrt{m}}\leq \Delta \varepsilon$. Similarly, the lower bound can be established by following the same line of proof, using Lemma~\ref{l22} together with the bound in \eqref{ww}. Consequently, the following quasi-isometric embedding holds with probability at least $1-2e^{c^{\prime}\varepsilon^2 m}$:
\begin{equation}
\left|\mathcal{D}^t\left(\mbX,\mbY\right)-\frac{1}{n_1n_2}\left\|\mbX-\mbY\right\|_1\right|\leq \varepsilon \left\|\mbX-\mbY\right\|_{\mathrm{F}}+c_3\left(|t|+\Delta \varepsilon\right),   
\end{equation}
for a universal constant $c_3>0$, provided that the number of samples $m$ satisfies $m\gtrsim \varepsilon^{-2} r\left(n_1+n_2\right)\log\left(1+\frac{\|\mathcal{K}^r\|}{\Delta\varepsilon}\right)$.
Setting $t=0$ recovers the result stated in Proposition~\ref{Theorem_q}.\\
$\bullet$ \textbf{Embedding result based on $\alpha$:}
We can reformulate the embedding result of Proposition~\ref{Theorem_q} in terms of the maximum-norm bound $\alpha$ rather than the Frobenius norm $\left\|\mbX-\mbY\right\|_{\mathrm{F}}$. Using the relation established in \eqref{37}, the concentration inequality in \eqref{43} can be restated in terms of $\alpha$ as follows:
\begin{equation}
\label{43_new}
\mathbb{P}\left(\sup_{\mbX^{\prime},\mbY^{\prime}\in\mathcal{K}^r_{\rho}}\left|\mathcal{D}^{t}\left(\mbX^{\prime},\mbY^{\prime}\right)-\frac{1}{n_1 n_2}\left\|\mbX^{\prime}-\mbY^{\prime}\right\|_1\right|\geq c_1|t|+\varepsilon\left(\Delta+|t|+\alpha\right)\right)\leq2e^{-c^\prime \varepsilon^2 m},
\end{equation}
for some universal constants $c_1,c^\prime>0$, provided that the number of samples $m$ satisfies \eqref{sample_m}.
Based on this result, we have, with probability at least $1-2e^{-c^{\prime}\varepsilon^2 m}$,
\begin{equation}
\begin{aligned}
\label{vv_new}
\left|\mathcal{D}^{(t-\rho\sqrt{P})}\left(\mbX^{\prime},\mbY^{\prime}\right)-\frac{1}{n_1n_2}\left\|\mbX^{\prime}-\mbY^{\prime}\right\|_1\right|\leq c_1\left|t-\rho\sqrt{P}\right|+\varepsilon\left(\Delta+\left|t-\rho\sqrt{P}\right|+\alpha\right),  
\end{aligned}
\end{equation}
and 
\begin{equation}
\label{ww_new}
\begin{aligned}
\left|\mathcal{D}^{(t+\rho\sqrt{P})}\left(\mbX^{\prime},\mbY^{\prime}\right)-\frac{1}{n_1n_2}\left\|\mbX^{\prime}-\mbY^{\prime}\right\|_1\right|\leq c_1\left|t+\rho\sqrt{P}\right|+\varepsilon\left(\Delta+\left|t+\rho\sqrt{P}\right|+\alpha\right),
\end{aligned}   
\end{equation}
whenever $m$ satisfies \eqref{sample_m}. Following the same reasoning as before and invoking Lemma~\ref{l22} together with \eqref{vv_new}, and assuming $\varepsilon<1$, there exists a constant $c_2>0$ such that
\begin{equation}
\mathcal{D}^t\left(\mbX,\mbY\right)-\frac{1}{n_1n_2}\left\|\mbX-\mbY\right\|_1\leq \varepsilon\left(\alpha+\Delta\right)+c_2\left(|t|+\rho\sqrt{P}+\frac{\rho}{\sqrt{n_1n_2}}+\frac{\Delta}{mP}+\frac{\rho}{m\sqrt{P}}\right).
\end{equation}
Selecting the parameters such that $m P = \varepsilon^{-1}>1$ and $\rho=\Delta\varepsilon^{\frac{3}{2}}\sqrt{m}$, which gives $\rho\sqrt{P}=\Delta\varepsilon$, $\frac{\rho}{\sqrt{n_1n_2}}\leq\Delta\varepsilon^{\frac{3}{2}}$, and $\frac{\rho}{m \sqrt{P}}=\Delta \varepsilon^2$. Following the same line of reasoning, the corresponding lower bound can be derived by invoking Lemma~\ref{l22} together with the bound in \eqref{ww_new}.

\section{Proof of Theorem~\ref{recovery}}
\label{ape2}
Under the consistency assumption and the quasi-isometric quantized embedding established in Proposition~\ref{Theorem_q}, by letting $\mathcal{D}$ approach zero and substituting the Frobenius-norm upper bound with the maximum-norm bound $\alpha$, as given in \eqref{37}, we obtain, with probability at least $1-2e^{-c^{\prime}\varepsilon^2 m}$,
\begin{equation}
\frac{1}{n_1n_2}\left\|\mbX-\mbY\right\|_1\lesssim \varepsilon\left(\alpha+\Delta\right),
\end{equation}
thereby completing the proof.

\section{Proof of Theorem~\ref{thr_onebit}}
\label{ape3}
We begin the proof by presenting the following lemma:
\begin{lemma}
\label{lem_1}
In the setting of Definition~\ref{def_T}, for a positive constant $c$ and $\epsilon\in(0,1)$, the following concentration bound holds:
\begin{equation}
\label{c1}
\mathbb{P}\left(\sup_{\mbX\in\mathcal{K}^r}\left|T\left(\mbX\right)-\frac{\zeta}{2}-\frac{\|\mbX\|_{\mathrm{F}}^2}{2\zeta n_1n_2}\right|\leq\epsilon\right)\geq 1-3e^{-c\epsilon^2m},
\end{equation}
provided that the number of observed samples satisfies
\begin{equation}
\label{c2}
m\gtrsim\epsilon^{-2}r(n_1+n_2)\log\left(1+\epsilon^{-1}+\|\mathcal{K}^r\|\right).
\end{equation}
\end{lemma}
\begin{IEEEproof}
For simplicity of notation, denote $d_{i,j}$ in Definition~\ref{def_T} by $d=\left|X-\tau\right|$. Following $\zeta\geq\alpha$, we can then write
\begin{equation}
\label{c3}
\begin{aligned}
\mathbb{E}_{\tau}d&=\frac{1}{2\zeta}\int_{-\zeta}^{\zeta}\left|X-\tau\right|\,d\tau\\&=\frac{1}{2\zeta}\left[\int_{-\zeta}^{X}X-\tau\,d\tau+\int_{X}^{\zeta}\tau-X\,d\tau\right]=\frac{\zeta}{2}+\frac{X^2}{2\zeta}.
\end{aligned}
\end{equation}
Therefore, we have
\begin{equation}
\label{c4}
\begin{aligned}
\mathbb{E}_{\tau}T\left(\mbX\right)&=\frac{1}{m}\sum_{(i,j)\in\Omega}\frac{\zeta}{2}+\frac{X_{i,j}^2}{2\zeta}\\&=\frac{\zeta}{2}+\frac{\|\mathcal{P}_{\Omega}\left(\mbX\right)\|_{\mathrm{F}}^2}{2\zeta m}.
\end{aligned}
\end{equation}
Computing the expected value of \eqref{c4} respect to the randomness of $(i,j)\in\Omega$ leads to
\begin{equation}
\label{c5}
\mathbb{E}_{\tau,(i,j)}T\left(\mbX\right)=\frac{\zeta}{2}+\frac{\|\mbX\|_{\mathrm{F}}^2}{2\zeta n_1n_2}.
\end{equation}
Note that for each random variable $d_{i,j}$, we have $0\leq d_{i,j}\leq\alpha+\zeta$. Then, by applying Lemma~\ref{hoff}, there exists a universal constant $c_1>0$ such that for a fixed matrix $\mbX\in\mathcal{K}^r$, the following concentration bound holds:
\begin{equation}
\label{c6}
\mathbb{P}\left(\left|T\left(\mbX\right)-\frac{\zeta}{2}-\frac{\|\mbX\|_{\mathrm{F}}^2}{2\zeta n_1n_2}\right|\geq t\right)\leq 2e^{-c_1t^2m}.
\end{equation}
To extend the result in \eqref{c6} to hold uniformly for all $\mbX\in\mathcal{K}^r$, we employ a standard covering argument. Specifically, we first approximate the set $\mathcal{K}^r$ by constructing a $\rho$-net $\mathcal{K}^r_\rho$ of $\mathcal{K}^r$. Then, leveraging the concentration inequality in \eqref{c6}, we apply a union bound over all matrices $\mbX\in\mathcal{K}^r_\rho$ to obtain a uniform guarantee. For the approximation step, by construction of the $\rho$-net $\mathcal{K}^r_\rho$ of $\mathcal{K}^r$, for any $\mbX\in\mathcal{K}^r$ there exists a matrix $\mbX^\prime\in\mathcal{K}^r_\rho$ such that $\left\|\mbX-\mbX^{\prime}\right\|_{\mathrm{F}}\leq\rho$. This allows us to express the following
\begin{equation}
\label{c7}
\begin{aligned}
\left|T\left(\mbX\right)-\frac{\zeta}{2}-\frac{\|\mbX\|_{\mathrm{F}}^2}{2\zeta n_1n_2}\right|&\leq\left|T\left(\mbX^{\prime}\right)-\frac{\zeta}{2}-\frac{\|\mbX^{\prime}\|_{\mathrm{F}}^2}{2\zeta n_1n_2}\right|+\left|T\left(\mbX\right)-T\left(\mbX^{\prime}\right)\right|+\frac{1}{2\zeta n_1n_2}\left|\|\mbX^{\prime}\|_{\mathrm{F}}^2-\|\mbX\|_{\mathrm{F}}^2\right|\\&\leq\left|T\left(\mbX^{\prime}\right)-\frac{\zeta}{2}-\frac{\|\mbX^{\prime}\|_{\mathrm{F}}^2}{2\zeta n_1n_2}\right|+\left|T\left(\mbX\right)-T\left(\mbX^{\prime}\right)\right|+\frac{1}{2\zeta n_1n_2}\left\|\mbX-\mbX^\prime\right\|_{\mathrm{F}}\left[\left\|\mbX\right\|_{\mathrm{F}}+\left\|\mbX^\prime\right\|_{\mathrm{F}}\right]\\&\leq\left|T\left(\mbX^{\prime}\right)-\frac{\zeta}{2}-\frac{\|\mbX^{\prime}\|_{\mathrm{F}}^2}{2\zeta n_1n_2}\right|+\underbrace{\left|T\left(\mbX\right)-T\left(\mbX^{\prime}\right)\right|}_{\text{Term}\hspace{1pt}\mathrm{\rom{1}}}+\frac{\alpha}{\zeta}\frac{\rho}{\sqrt{n_1n_2}}.
\end{aligned}
\end{equation}
Based on the definition of the mapping $T(\cdot)$ in \eqref{def_T}, we can upper bound $\text{Term}\hspace{1pt}\mathrm{\rom{1}}$ in \eqref{c7} as follows:
\begin{equation}
\label{c8}
\begin{aligned}
\text{Term}\hspace{1pt}\mathrm{\rom{1}}&=\frac{1}{m}\left|\sum_{(i,j)\in\Omega}\left[\left|X_{i,j}-\tau_{i,j}\right|-\left|X^{\prime}_{i,j}-\tau_{i,j}\right|\right]\right|\\&\leq\frac{1}{m}\sum_{(i,j)\in\Omega}\left|\left|X_{i,j}-\tau_{i,j}\right|-\left|X^{\prime}_{i,j}-\tau_{i,j}\right|\right|\\&\leq\underbrace{\frac{1}{m}\sum_{(i,j)\in\Omega}\left|X_{i,j}-X^{\prime}_{i,j}\right|}_{\text{Term}\hspace{1pt}\mathrm{\rom{2}}}.
\end{aligned}
\end{equation}
To upper bound $\text{Term}\hspace{1pt}\mathrm{\rom{2}}$ with high probability, we employ the following lemma:
\begin{lemma}
\label{lem_2}
Consider the set 
\begin{equation}
\label{c9}
\mathcal{X}^r=\left\{\mbU\in\mathbb{R}^{n_1\times n_2}\mid\operatorname{rank}\left(\mbU\right)\leq r,\left\|\mbU\right\|_{\mathrm{F}}\leq1\right\}.
\end{equation}
Define $K(\mbU)=\frac{1}{m}\left\|\mathcal{P}_{\Omega}\left(\mbU\right)\right\|_1$. Then, for some positive universal constant $c_2$ and $\delta\in(0,1)$, the following universal bound holds:
\begin{equation}
\label{c10}
\mathbb{P}\left(\sup_{\mbU\in\mathcal{X}^r}\left[K(\mbU)-\frac{\left\|\mbU\right\|_1}{n_1n_2}\right]\geq\delta\left\|\mbU\right\|_{\mathrm{F}}\right)\leq e^{-c_2\delta^2m},
\end{equation}
provided that the number of samples satisfies
\begin{equation}
\label{c11}
m\gtrsim\delta^{-2}r(n_1+n_2)\log\left(1+\left\|\mathcal{X}^r\right\|\left(1+\frac{\delta^{-1}}{\sqrt{n_1n_2}}\right)\right).
\end{equation}
\end{lemma}
The proof of Lemma~\ref{lem_2} is provided in Appendix~\ref{ape5}. Let $\mbE=\frac{\mbX-\mbX^\prime}{\rho}$. Since $\mbX-\mbX^\prime\in\left(\mathcal{K}^{r}-\mathcal{K}^r\right)\bigcap \rho \mathbb{B}_{\mathrm{F}}^{n_1\times n_2}$, it follows that $\mbE\in\mathcal{X}^{2r}$. By applying Lemma~\ref{lem_2}, we obtain the following upper bound for $\text{Term}\hspace{1pt}\mathrm{\rom{2}}$ with failure probability of at most $e^{-c_2\delta^2m}$.
\begin{equation}
\label{c12}
\begin{aligned}
\text{Term}\hspace{1pt}\mathrm{\rom{2}}&=\frac{\rho}{m}\sum_{(i,j)\in\Omega}\left|\frac{X_{i,j}-X^{\prime}_{i,j}}{\rho}\right|\\&=\frac{\rho}{m}\sum_{(i,j)\in\Omega}\left|E_{i,j}\right|\\&\leq\rho\left(\frac{\left\|\mbE\right\|_1}{n_1n_2}+\delta\left\|\mbE\right\|_{\mathrm{F}}\right)\\&\leq\frac{\left\|\mbX-\mbX^\prime\right\|_{\mathrm{F}}}{\sqrt{n_1n_2}}+\delta\left\|\mbX-\mbX^\prime\right\|_{\mathrm{F}}\\&\leq\left(\delta+\frac{1}{\sqrt{n_1n_2}}\right)\rho,
\end{aligned}
\end{equation}
if the number of samples $m$ satisfies \eqref{c11}. By combining \eqref{c12} with the result in \eqref{c7}, we obtain
\begin{equation}
\label{c13}
\left|T\left(\mbX\right)-\frac{\zeta}{2}-\frac{\|\mbX\|_{\mathrm{F}}^2}{2\zeta n_1n_2}\right|\leq\left|T\left(\mbX^{\prime}\right)-\frac{\zeta}{2}-\frac{\|\mbX^{\prime}\|_{\mathrm{F}}^2}{2\zeta n_1n_2}\right|+\left(\delta+\frac{2}{\sqrt{n_1n_2}}\right)\rho,
\end{equation}
where we have used the fact that $\zeta\geq\alpha$. Taking the supremum over all $\mbX\in\mathcal{K}^r$ then leads to
\begin{equation}
\label{c14}
\sup_{\mbX\in\mathcal{K}^r}\left|T\left(\mbX\right)-\frac{\zeta}{2}-\frac{\|\mbX\|_{\mathrm{F}}^2}{2\zeta n_1n_2}\right|\leq\underbrace{\sup_{\mbX^\prime\in\mathcal{K}^r_\rho}\left|T\left(\mbX^{\prime}\right)-\frac{\zeta}{2}-\frac{\|\mbX^{\prime}\|_{\mathrm{F}}^2}{2\zeta n_1n_2}\right|}_{\text{Term}\hspace{1pt}\mathrm{\rom{3}}}+\left(\delta+\frac{2}{\sqrt{n_1n_2}}\right)\rho.
\end{equation}
To control $\text{Term}\hspace{1pt}\mathrm{\rom{3}}$, we apply the concentration bound from \eqref{c6} via a union bound:
\begin{equation}
\label{c15}
\mathbb{P}\left(\sup_{\mbX^\prime\in\mathcal{K}^r_\rho}\left|T\left(\mbX^\prime\right)-\frac{\zeta}{2}-\frac{\|\mbX^\prime\|_{\mathrm{F}}^2}{2\zeta n_1n_2}\right|\geq t\right)\leq 2e^{\mathcal{H}\left(\mathcal{K}^r,\rho\right)-c_1t^2m}.
\end{equation}
We can equivalently rewrite the above expression so that the concentration inequality holds with a failure probability of at most $2e^{-c_3t^2m}$, for some universal constant $c_3>0$, provided that $m\gtrsim t^{-2}\mathcal{H}\left(\mathcal{K}^r,\rho\right)$. Combining this result with \eqref{c14}, we obtain that the following bound holds with probability at least $1-2e^{-c_3t^2m}-e^{-c_2\delta^2m}$
\begin{equation}
\label{c16}
\sup_{\mbX\in\mathcal{K}^r}\left|T\left(\mbX\right)-\frac{\zeta}{2}-\frac{\|\mbX\|_{\mathrm{F}}^2}{2\zeta n_1n_2}\right|\leq t+\left(\delta+\frac{2}{\sqrt{n_1n_2}}\right)\rho,
\end{equation}
if the number of samples $m$ satisfies
\begin{equation}
\label{c17}
m\gtrsim\max\left(\delta^{-2}r(n_1+n_2)\log\left(1+\left\|\mathcal{X}^r\right\|\left(1+\frac{\delta^{-1}}{\sqrt{n_1n_2}}\right)\right),t^{-2}r(n_1+n_2)\log\left(1+\frac{\left\|\mathcal{K}^r\right\|}{\rho}\right)\right).
\end{equation}
To derive \eqref{c17}, we used the upper bound of $\mathcal{H}\left(\mathcal{K}^r,\rho\right)$ given in \eqref{kol_bound}. To obtain the desired result in \eqref{c1}, we should set $\epsilon=t+\left(\delta+\frac{2}{\sqrt{n_1n_2}}\right)\rho$. Accordingly, to complete the proof, we set $t=\frac{\epsilon}{2}$, $\delta=\frac{\epsilon}{2}$, and $\rho=\frac{\epsilon\sqrt{n_1n_2}}{\epsilon\sqrt{n_1n_2}+4}$. With this selection, the result of Lemma~\ref{lem_2} follows, with the required sample complexity $m$ and probability at least $1-3e^{-c\epsilon^2m}$ for some universal constant $c>0$.
\end{IEEEproof}
Based on the result of Lemma~\ref{lem_2}, we are now ready to prove Theorem~\ref{thr_onebit}. Let $\mbZ=\frac{1}{2}\left(\mbX+\mbY\right)$ for any $\mbX,\mbY\in\mathcal{K}^r$ that satisfy the consistent reconstruction property in Definition~\ref{def_2}. By the definition of $\mbZ$, for every $(i,j)\in\Omega$, we can write $Z_{i,j}-\tau_{i,j}=\frac{1}{2}\left(X_{i,j}-\tau_{i,j}+Y_{i,j}-\tau_{i,j}\right)$. Since $\mbX$ and $\mbY$ belong to the same consistency cell, it follows that
\begin{equation}
\label{c18}
\left|Z_{i,j}-\tau_{i,j}\right|=\frac{1}{2}\left(\left|X_{i,j}-\tau_{i,j}\right|+\left|Y_{i,j}-\tau_{i,j}\right|\right),~(i,j)\in\Omega.
\end{equation}
Taking the empirical average of both sides of \eqref{c18} yields
\begin{equation}
\label{c19}
T\left(\mbZ\right)=\frac{1}{2}\left[T\left(\mbX\right)+T\left(\mbY\right)\right].
\end{equation}
Since $\mbZ\in\frac{1}{2}\left(\mathcal{K}^r+\mathcal{K}^r\right)$, it follows that $\operatorname{rank}(\mbZ)\leq2r$ and $\|\mbZ\|_{\mathrm{max}}\leq\alpha$.
By Lemma~\ref{lem_1}, if the sampling condition in \eqref{c2} holds, then with failure probability at most $3e^{-c\epsilon^2m}$, we have
\begin{equation}
\label{c20}
\frac{\|\mbZ\|_{\mathrm{F}}^2}{2\zeta n_1n_2}\geq T(\mbZ)-\frac{\zeta}{2}-\epsilon.
\end{equation}
Combining \eqref{c19} and \eqref{c20} gives
\begin{equation}
\label{c21}
\begin{aligned}
\frac{\|\mbZ\|_{\mathrm{F}}^2}{2\zeta n_1n_2}&\geq\frac{1}{2}\left[T(\mbX)+T(\mbY)\right]-\frac{\zeta}{2}-\epsilon\\&\geq\frac{1}{2}\left[\frac{\|\mbX\|_{\mathrm{F}}^2}{2\zeta n_1n_2}+\frac{\zeta}{2}-\epsilon+\frac{\|\mbY\|_{\mathrm{F}}^2}{2\zeta n_1n_2}+\frac{\zeta}{2}-\epsilon\right]-\frac{\zeta}{2}-\epsilon\\&=\frac{1}{4\zeta n_1n_2}\left[\|\mbX\|_{\mathrm{F}}^2+\|\mbY\|_{\mathrm{F}}^2\right]-2\epsilon.
\end{aligned}
\end{equation}
By substituting $\mbZ=\frac{1}{2}\left(\mbX+\mbY\right)$ into \eqref{c21}, we can express the result in terms of $\mbX$ and $\mbY$:
\begin{equation}
\label{c22}
\|\mbX+\mbY\|_{\mathrm{F}}^2\geq2\left(\|\mbX\|_{\mathrm{F}}^2+\|\mbY\|_{\mathrm{F}}^2\right)-16\epsilon\zeta n_1n_2.
\end{equation}
Finally, applying the parallelogram law yields
\begin{equation}
\label{c23}
\begin{aligned}
\|\mbX-\mbY\|_{\mathrm{F}}^2&=2\left[\|\mbX\|_{\mathrm{F}}^2+\|\mbY\|_{\mathrm{F}}^2\right]-\|\mbX+\mbY\|_{\mathrm{F}}^2\\&\leq16\epsilon\zeta n_1n_2.
\end{aligned}
\end{equation}
This completes the proof of Theorem~\ref{thr_onebit}.

\section{Proof of Theorem~\ref{thr_2}}
\label{ape4}
Let $\mbZ=\frac{1}{2}(\mbX+\mbY)$ for any $\mbX,\mbY\in\mathcal{K}^r$ that do not necessarily satisfy the consistent reconstruction property in Definition~\ref{def_2}. By definition, for every $(i,j)\in\Omega$, we have $Z_{i,j}-\tau_{i,j}=\frac{1}{2}\left(X_{i,j}-\tau_{i,j}+Y_{i,j}-\tau_{i,j}\right)$. Define the set
\begin{equation}
\label{p_1}
\mathcal{G}=\left\{(i,j)\in\Omega:~\operatorname{sgn}\left(X_{i,j}-\tau_{i,j}\right)\neq\operatorname{sgn}\left(Y_{i,j}-\tau_{i,j}\right)\right\}.
\end{equation}
For all $(i,j)\in\Omega\setminus\mathcal{G}$, i.e., indices where $\mbX$ and $\mbY$ are consistent in sign, we have
\begin{equation}
\label{p_2}
\left|Z_{i,j}-\tau_{i,j}\right|=\frac{1}{2}\left(\left|X_{i,j}-\tau_{i,j}\right|+\left|Y_{i,j}-\tau_{i,j}\right|\right),
\end{equation}
which follows directly from the consistent reconstruction property. Conversely, for all $(i,j)\in\mathcal{G}$, where $\mbX$ and $\mbY$ do not have the same sign, it can be readily verified that
\begin{equation}
\label{p_3}
\begin{aligned}
\left|Z_{i,j}-\tau_{i,j}\right|&=\frac{1}{2}\left(\left|X_{i,j}-\tau_{i,j}\right|+\left|Y_{i,j}-\tau_{i,j}\right|\right)\\&-\min\left(\left|X_{i,j}-\tau_{i,j}\right|,\left|Y_{i,j}-\tau_{i,j}\right|\right).
\end{aligned}
\end{equation}
Taking the empirical average over all $(i,j)\in\Omega$ yields
\begin{equation}
\label{p_4}
T(\mbZ)=\frac{1}{2}\left[T(\mbX)+T(\mbY)\right]-R,
\end{equation}
where
\begin{equation}
\label{p_5}
R=\frac{1}{m}\sum_{(i,j)\in\mathcal{G}}\min\left(\left|X_{i,j}-\tau_{i,j}\right|,\left|Y_{i,j}-\tau_{i,j}\right|\right).
\end{equation}
According to Lemma~\ref{lem_1}, if the sampling condition in \eqref{c2} holds, then with failure probability at most $3e^{-c\epsilon^2m}$ we have
\begin{equation}
\label{p_6}
\frac{\|\mbZ\|_{\mathrm{F}}^2}{2\zeta n_1n_2}\geq T(\mbZ)-\frac{\zeta}{2}-\epsilon,
\end{equation}
which together with \eqref{p_4} results in
\begin{equation}
\label{p_7}
\begin{aligned}
\frac{\|\mbZ\|_{\mathrm{F}}^2}{2\zeta n_1n_2}&\geq\frac{1}{2}\left[T(\mbX)+T(\mbY)\right]-R-\frac{\zeta}{2}-\epsilon\\&\geq\frac{1}{2}\left[\frac{\|\mbX\|_{\mathrm{F}}^2}{2\zeta n_1n_2}+\frac{\zeta}{2}-\epsilon+\frac{\|\mbY\|_{\mathrm{F}}^2}{2\zeta n_1n_2}+\frac{\zeta}{2}-\epsilon\right]-R-\frac{\zeta}{2}-\epsilon\\&=\frac{1}{4\zeta n_1n_2}\left[\|\mbX\|_{\mathrm{F}}^2+\|\mbY\|_{\mathrm{F}}^2\right]-R-2\epsilon.
\end{aligned}
\end{equation}
Based on the definition of $\mbZ$ and the parallelogram law, we obtain
\begin{equation}
\label{p_8}
\|\mbX-\mbY\|_{\mathrm{F}}^2\leq8\zeta n_1n_2R+16\zeta n_1n_2\epsilon.
\end{equation}
The value of $R$ can be bounded as
\begin{equation}
\label{p_9}
R\leq\frac{|\mathcal{G}|}{m}2\zeta,
\end{equation}
which, when substituted into \eqref{p_8}, yields
\begin{equation}
\label{p_10}
\|\mbX-\mbY\|_{\mathrm{F}}\leq4\sqrt{\epsilon\zeta n_1n_2+\zeta^2\frac{|\mathcal{G}|}{m}n_1n_2},
\end{equation}
thereby completing the proof.

\section{Proof of Lemma~\ref{lem_2}}
\label{ape5}
It is easy to observe that 
\begin{equation}
\label{d1}
\mathbb{E}_{(i,j)}K(\mbU)=\frac{1}{m}\mathbb{E}_{(i,j)}\left\|\mathcal{P}_{\Omega}\left(\mbU\right)\right\|_1=\frac{\left\|\mbU\right\|_1}{n_1n_2}.
\end{equation}
Since $\mbU\in\mathcal{X}^r$, then each $U_{i,j}$ is a subgaussian random variable. Then, for some universal constant $c_1$, we can provide the following one-sided concentration bound for a fixed $\mbU\in\mathcal{X}^r$ following Lemma~\ref{hoff}:
\begin{equation}
\label{d2}
\mathbb{P}\left(K(\mbU)-\frac{\left\|\mbU\right\|_1}{n_1n_2}\geq t\left\|\mbU\right\|_{\mathrm{F}}\right)\leq e^{-c_1t^2m}.
\end{equation}
To extend this result to all $\mbU \in \mathcal{X}^r$, we follow a similar argument to that used in the proof of Lemma~\ref{lem_1}. Specifically, we first approximate the set $\mathcal{X}^r$ by constructing a $\rho$-net $\mathcal{K}^r_\rho$. Then, leveraging the concentration bound in \eqref{d2}, we apply a union bound over all matrices $\mbU \in \mathcal{X}^r_\rho$. In the approximation step, for any $\mbU \in \mathcal{X}^r$, there exists a matrix $\mbU^\prime \in \mathcal{K}^r_\rho$ such that $\left\|\mbU - \mbU^\prime\right\|_{\mathrm{F}} \leq \rho$. This enables us to write the following:
\begin{equation}
\label{d3}
\begin{aligned}
K(\mbU)-\frac{\left\|\mbU\right\|_1}{n_1n_2}&=\left(K(\mbU^\prime)-\frac{\left\|\mbU^\prime\right\|_1}{n_1n_2}\right)+\left(K(\mbU)-K(\mbU^\prime)\right)+\frac{1}{n_1n_2}\left(\left\|\mbU^\prime\right\|_1-\left\|\mbU\right\|_1\right)\\&\leq\left(K(\mbU^\prime)-\frac{\left\|\mbU^\prime\right\|_1}{n_1n_2}\right)+\left(K(\mbU)-K(\mbU^\prime)\right)+\frac{1}{n_1n_2}\left\|\mbU-\mbU^\prime\right\|_1\\&\leq\left(K(\mbU^\prime)-\frac{\left\|\mbU^\prime\right\|_1}{n_1n_2}\right)+\underbrace{\left(K(\mbU)-K(\mbU^\prime)\right)}_{\text{Term}\hspace{1pt}(a)}+\frac{\rho}{\sqrt{n_1n_2}}.
\end{aligned}
\end{equation}
Based on the definition of $K(\cdot)$, we can upper bound $\text{Term}\hspace{1pt}(a)$ in \eqref{d3} as follows:
\begin{equation}
\label{d4}
\begin{aligned}
\text{Term}\hspace{1pt}(a)&=\frac{1}{m}\sum_{(i,j)\in\Omega}\left(\left|U_{i,j}\right|-\left|U_{i,j}^\prime\right|\right)\\&\leq\frac{1}{m}\sum_{(i,j)\in\Omega}\left|U_{i,j}-U_{i,j}^\prime\right|\\&=K(\mbU-\mbU^\prime).
\end{aligned}
\end{equation}
Let us introduce the smallest constant $\delta>0$ such that
\begin{equation}
\label{d5}
K(\mbU)-\frac{\left\|\mbU\right\|_1}{n_1n_2}\leq\delta\left\|\mbU\right\|_{\mathrm{F}},~\forall\mbU\in\mathbb{R}^{n_1\times n_2}~\text{with}~\operatorname{rank}\left(\mbU\right)\leq r.
\end{equation}
Since $\mbU-\mbU^\prime$ has rank at most $2r$, it can be written as $\mbU-\mbU^\prime=\mbM+\mbN$, where both $\mbM,\mbN\in\mathbb{R}^{n_1\times n_2}$ have rank at most $r$, and $\left\|\mbU-\mbU^\prime\right\|_{\mathrm{F}}^2=\left\|\mbM\right\|_{\mathrm{F}}^2+\left\|\mbN\right\|_{\mathrm{F}}^2$. Then
\begin{equation}
\label{d6}
\begin{aligned}
\text{Term}\hspace{1pt}(a)&\leq K(\mbU-\mbU^\prime)\\&=K(\mbM+\mbN)\\&\leq K(\mbM)+K(\mbN)\\&\leq\frac{\left\|\mbM\right\|_1}{n_1n_2}+\delta\left\|\mbM\right\|_{\mathrm{F}}+\frac{\left\|\mbN\right\|_1}{n_1n_2}+\delta\left\|\mbN\right\|_{\mathrm{F}}\\&\leq\frac{\left\|\mbM\right\|_{\mathrm{F}}+\left\|\mbN\right\|_{\mathrm{F}}}{\sqrt{n_1n_2}}+\delta\left(\left\|\mbM\right\|_{\mathrm{F}}+\left\|\mbN\right\|_{\mathrm{F}}\right)\\&\leq\sqrt{\frac{2}{n_1n_2}}\sqrt{\left\|\mbM\right\|_{\mathrm{F}}^2+\left\|\mbN\right\|_{\mathrm{F}}^2}+\sqrt{2}\delta\sqrt{\left\|\mbM\right\|_{\mathrm{F}}^2+\left\|\mbN\right\|_{\mathrm{F}}^2}\\&=\sqrt{\frac{2}{n_1n_2}}\left\|\mbU-\mbU^\prime\right\|_{\mathrm{F}}+\sqrt{2}\delta\left\|\mbU-\mbU^\prime\right\|_{\mathrm{F}}\\&\leq\sqrt{\frac{2}{n_1n_2}}\rho+\sqrt{2}\delta\rho.
\end{aligned}
\end{equation}
By combining \eqref{d6} with the result in \eqref{d3}, we obtain
\begin{equation}
\label{d7}
K(\mbU)-\frac{\left\|\mbU\right\|_1}{n_1n_2}\leq\left(K(\mbU^\prime)-\frac{\left\|\mbU^\prime\right\|_1}{n_1n_2}\right)+\frac{1+\sqrt{2}}{\sqrt{n_1n_2}}\rho+\sqrt{2}\delta\rho.
\end{equation}
Taking the supremum over all $\mbU\in\mathcal{X}^r$ then yields
\begin{equation}
\label{d8}
\sup_{\mbU\in\mathcal{X}^r}\left[K(\mbU)-\frac{\left\|\mbU\right\|_1}{n_1n_2}\right]\leq\underbrace{\sup_{\mbU^\prime\in\mathcal{X}^r_\rho}\left[K(\mbU^\prime)-\frac{\left\|\mbU^\prime\right\|_1}{n_1n_2}\right]}_{\text{Term}\hspace{1pt}(b)}+\frac{1+\sqrt{2}}{\sqrt{n_1n_2}}\rho+\sqrt{2}\delta\rho.
\end{equation}
To control $\text{Term}\hspace{1pt}(b)$, we invoke the concentration inequality in \eqref{d2} together with a union bound,
\begin{equation}
\label{d9}
\mathbb{P}\left(\sup_{\mbU^\prime\in\mathcal{X}^r_\rho}\left[K(\mbU^\prime)-\frac{\left\|\mbU^\prime\right\|_1}{n_1n_2}\right]\geq t\left\|\mbU^\prime\right\|_{\mathrm{F}}\right)\leq e^{\mathcal{H}\left(\mathcal{X}^r,\rho\right)-c_1t^2m}.
\end{equation}
Combining this result with \eqref{d8}, we obtain that the following bound holds with failure probability at most $2e^{-c_2t^2m}$, for some universal constant $c_2$:
\begin{equation}
\label{d10}
\sup_{\mbU\in\mathcal{X}^r}\left[K(\mbU)-\frac{\left\|\mbU\right\|_1}{n_1n_2}\right]\leq t+\frac{1+\sqrt{2}}{\sqrt{n_1n_2}}\rho+\sqrt{2}\delta\rho,
\end{equation}
provided that $m\gtrsim t^{-2}\mathcal{H}\left(\mathcal{X}^r,\rho\right)$. Following our assumption in \eqref{d5}, it must then hold that
\begin{equation}
\label{d11}
\delta\leq t+\frac{1+\sqrt{2}}{\sqrt{n_1n_2}}\rho+\sqrt{2}\delta\rho.
\end{equation}
Finally, by setting $t=\frac{\delta}{2}$ and $\rho=\frac{\delta/2\sqrt{n_1n_2}}{1+\sqrt{2}+\sqrt{2}\delta\sqrt{n_1n_2}}$, the proof is complete.

\bibliographystyle{IEEEtran}
\bibliography{references}

\end{document}